\documentclass[journal,comsoc]{IEEEtran}
\usepackage{longtable}
\usepackage{graphicx}
\usepackage{subfigure}
\usepackage{epstopdf}
\usepackage{citesort}
\usepackage{bm}
\usepackage{amssymb}
\usepackage{amsthm}
\usepackage{amsfonts}
\usepackage{amsmath}
\usepackage{epsfig}
\usepackage{fancybox}
\usepackage{textcomp}
\usepackage{multirow}
\usepackage{booktabs}
\usepackage{mathrsfs}
\usepackage{pifont}
\usepackage{bbm}
\usepackage[ruled]{algorithm2e}
\usepackage{euscript}
\usepackage{xcolor}

\newtheorem{theorem}{Theorem}
\newtheorem{lemma}{Lemma}
\newtheorem{Scale_property}{Scale Property}

\makeatletter
\def\ScaleIfNeeded{%
\ifdim\Gin@nat@width>\linewidth \linewidth \else \Gin@nat@width
\fi } \makeatother

\begin{document}

 \title{{ Spectral and Energy Efficiency of Uplink D2D Underlaid Massive MIMO Cellular Networks}}
 \author{ Anqi He, Lifeng Wang,~\IEEEmembership{Member,~IEEE,}  Yue Chen,~\IEEEmembership{Senior Member,~IEEE,}\\ Kai-Kit Wong,~\IEEEmembership{Fellow,~IEEE}, and Maged Elkashlan,~\IEEEmembership{Member,~IEEE}
\thanks{A. He, Y. Chen, and M. Elkashlan are with School of Electronic Engineering and Computer
Science, Queen Mary University of London, London, UK (Email: $\rm\{a.he,yue.chen,maged.elkashlan\}@qmul.ac.uk$).}
\thanks{L. Wang and K.-K. Wong are with the Department of Electronic and Electrical Engineering, University College London, London, UK (Email: $\rm\{lifeng.wang, kai$-$\rm kit.wong\}@ucl.ac.uk$).}
}

\maketitle

\begin{abstract}
One of key 5G scenarios is that device-to-device (D2D) and massive multiple-input multiple-output (MIMO) will be co-existed. However, interference in the uplink D2D underlaid massive MIMO cellular networks needs to be coordinated, due to the vast cellular and D2D transmissions. To this end, this paper introduces a spatially dynamic power control solution for mitigating the cellular-to-D2D and D2D-to-cellular interference. In particular, the proposed D2D power control policy is rather flexible including the special cases of no D2D links or using maximum transmit power. Under the considered power control, an analytical approach is developed to evaluate the spectral efficiency (SE) and energy efficiency (EE) in such networks. Thus, the exact expressions of SE for a cellular user or D2D transmitter are derived, which quantify the impacts of key system parameters such as massive MIMO antennas and D2D density. Moreover, the D2D scale properties are obtained, which provide the sufficient conditions for achieving the anticipated SE. Numerical results corroborate our analysis and show that the proposed power control solution can efficiently mitigate interference between the cellular and D2D tier. The results demonstrate that there exists the optimal D2D density for maximizing the area SE of D2D tier. In addition, the achievable EE of a cellular user can be comparable to that of a D2D user.
\end{abstract}


\section{Introduction}
With the increasing demand for high-definition mobile multimedia and fast mobile internet services, fifth generation (5G) mobile networks are anticipated to support the deluge of data traffic~\cite{Jeffrey_5G}. According to 5G-PPP, one of the key performance indicators (KPIs) in 5G  mobile networks is that the energy consumption will be at least ten times lower than 2010~\cite{5GPPP_2015}, which means that energy efficiency (EE) will play an important role in the 5G design.  Among the emerging technologies~\cite{D_liu_survey}, massive multiple-input multiple-output (MIMO) and device-to-device (D2D) are viewed as two key enablers to achieve 5G targets.

Massive MIMO can drastically improve the spectral efficiency (SE) by using large number of antennas and accommodating dozens of users in the same radio channel~\cite{ngo2013energy}. {However, the circuit power consumption increases with the number of antennas, which may deteriorate the downlink EE of massive MIMO systems~\cite{Anqi_letter_2015}. The existing works such as~\cite{J_zhang_2016_Dai,Jiayi_JSAC_2017} have investigated the use of low-resolution/mixed analog-to-digital convertors (ADCs) in an attempt to reduce circuit power consumption.} D2D takes advantage of the proximity to support direct transmissions without the aid of base stations (BSs) or the core networks. As a result, D2D can improve both SE and EE, and decrease the delay~\cite{Asadi_A_2014}. However, the D2D distance plays a dominant role in D2D transmission, which significantly affects the D2D performance. When D2D users and cellular users share the same frequency bands in D2D underlaid massive MIMO cellular networks,  interference becomes a key issue to be addressed. In such networks, severe co-channel interference exists due to the following two key factors:
 \begin{itemize}
   \item In contrast to the traditional cellular networks, massive MIMO cellular networks enable much more cellular transmissions at the same time and frequency band. As such, the inter-cell interference and cellular-to-D2D interference will be much higher than ever before.
   \item  D2D users are expected to be dense for offloading the network traffic. As such, the D2D-to-cellular interference will significantly deteriorate the cellular transmissions.
 \end{itemize}
 Currently, interference mitigation in such networks remains an open problem.

 This paper focuses on uplink D2D underlaid massive MIMO cellular networks. In order to coordinate the inter-cell interference, cellular-to-D2D interference, and D2D-to-cellular interference, we consider two power control schemes for cellular users and D2D users, respectively. To date, there are few results available for presenting the uplink SE and EE with power control in such networks. Therefore, this paper reveals design insights into the interplay between massive MIMO and D2D in the uplink cellular setting.

\subsection{Related Works and Motivation}
The implementation of D2D in the cellular networks is a promising approach to offload cellular traffic and avoid congestion in the core network~\cite{D_Feng_2014}. In \cite{ElSawy_2014_TCOM}, D2D and cellular mode selection was considered for achieving better link quality. The work of \cite{Min_2011} assumed that D2D user has a protection zone such that the uplink cellular-to-D2D interference cannot be larger than a threshold, and showed that the capacity of a D2D link can be enhanced while the capacity loss of cellular users is negligible. In \cite{Yang_2016_TII}, cooperative transmissions in the D2D overlay/underlay cellular networks were studied, and it was verified that the D2D transmission capacity can be  enhanced with the assistance of relay.  In \cite{Ma_C_2016}, a contract-based cooperative spectrum sharing was developed to exploit the transmission opportunities for the D2D links and keep the maximum profit of the cellular links. Nevertheless, the aforementioned literature only considered D2D communications in the traditional cellular networks, and more research efforts are needed to comprehensively understand the D2D communications in the future cellular networks such as 5G with many disruptive technologies \cite{D_liu_survey}.

Power control has been widely studied in conventional D2D underlaid cellular networks for interference
management~\cite{5949138,K_S_Ali_2015,Huang_2016_TCOM,lee2015power,Zhong_Kit_2015,lin2016joint,Yang_huang_2016}. In \cite{5949138},  a dynamic power control mechanism was proposed
for controlling the D2D user's transmit power, so as to reduce the D2D-to-cellular interference. In \cite{K_S_Ali_2015}, the truncated channel inversion power control was adopted such that the data rate is constant during the transmissions, and  D2D and cellular users cannot transmit signals if their transmit power is larger than a predefined value. A centralized power control solution in D2D enabled two-tier cellular networks was proposed by \cite{Huang_2016_TCOM}.  In \cite{lee2015power},  power control algorithms were proposed for mitigating the cross-tier
interference between the D2D links and one single cellular link. In the work of \cite{lee2015power}, centralized power control problem was formulated
as a linear-fractional programming and the optimal solution was obtained by using standard convex programming tools. D2D power control in conventional uplink MIMO cellular networks was studied by authors of \cite{Zhong_Kit_2015}, where a distributed resource allocation algorithm  was proposed based on the game-theoretic model. In \cite{lin2016joint},
joint beamforming and power control was studied in a single cell consisting of one D2D pair and multiple cellular users, and the optimization
problem was formulated for minimizing the total transmit power. The work of \cite{Yang_huang_2016} also considered a D2D underlaid single cell network and  investigated the downlink power control for maximizing the sum rate of D2D pairs. However, these prior works only pay attention to power control problem in the conventional D2D underlaid cellular networks.  Moreover, the majority of the existing D2D power control designs such as \cite{Zhong_Kit_2015,lin2016joint,Yang_huang_2016} need the global channel state information (CSI), which is challenging in D2D underlaid  massive MIMO cellular networks, since the CSI  between the D2D users and massive MIMO enabled BSs cannot be easily obtained.

The opportunities and challenges of the co-existence of the massive MIMO and D2D have recently been investigated in the uplink~\cite{lin2015interplay} and downlink transmissions~\cite{shalmashi2015energy}. {In \cite{lin2015interplay}, D2D and massive MIMO aided cellular uplink SE were studied and the interplay between D2D and massive MIMO was exploited, which showed that there is a loss in cellular SE due to D2D underlay. To redeem the cellular performance loss, authors in  \cite{lin2015interplay} assumed that the number of canceled D2D interfering signals is scaled with the number of BS antennas. { In \cite{shalmashi2015energy}, downlink sum rate and EE were analyzed in a single massive MIMO cell}, where multiple D2D transmitters were randomly located. The work of \cite{shalmashi2015energy} utilized equal power allocation without considering interference management, and showed that the benefits of the coexistence of D2D and massive MIMO are limited by the density of D2D users.  Particularly when there are vast D2D links and each massive MIMO BS provides services for dozens of users, interference becomes a major issue and needs to be mitigated~\cite{lin2015interplay,shalmashi2015energy}.  Although the existing works \cite{lin2015interplay} and \cite{shalmashi2015energy} have respectively investigated the uplink and downlink features of the massive MIMO cellular networks with underlaid D2D, the interference management via power control in such networks has not been conducted yet. To date, no effort has been devoted
to analyze the effects of uplink power control on the SE and EE of the D2D underlaid  massive MIMO cellular networks.

\subsection{Contributions}
This paper focuses on the uplink D2D underlaid massive MIMO cellular networks, in which power control is adopted for interference coordination. The detailed contributions are summarized as follows:
\begin{itemize}
\item We consider a massive MIMO aided multi-cell network, where cellular users are associated with the nearest BS for uplink transmissions, and the D2D transmitters are randomly located. In such a network, we introduce a power control solution to mitigate the inter-cell, cellular-to-D2D and D2D-to-cellular interference. { Specifically, cellular users are recommended to utilize open-loop power control with maximum transmit power constraint, to mitigate the uplink inter-cell interference and cellular-to-D2D interference. Considering the fact that D2D transmissions are unpredictable, the rationale behind the proposed D2D power control policy is that the average received D2D signal power from an arbitrary D2D transmitter should be controlled at a certain level with maximum D2D transmit power constraint, to mitigate the D2D-to-cellular interference. Different from the existing designs such as \cite{Zhong_Kit_2015,lin2016joint,Yang_huang_2016,lin2015interplay}, the proposed D2D power control policy does not require the global CSI. In addition, the positions of D2D transmitters and BSs are modeled by independent Poisson point processes, which indicates that the transmit power of cellular user or D2D transmitter is spatially dynamic in this paper.}

\item We develop an analytical approach to quantify the impacts of massive MIMO and D2D. Based on the proposed power control polices, the exact expressions of SE for a cellular user or D2D transmitter are derived, which accounts for the features of massive MIMO and D2D. Since the severe interference resulted from dense D2D transmissions can drastically degrade the SE, we provide two D2D scale properties, which explicitly show that the D2D density should not be larger than a critical value for achieving the desired SE. The average power consumption under the proposed power control policies are derived, which helps us evaluate the EE in such networks. It is confirmed from the derived results that adding more massive MIMO antennas can enhance both SE and EE of a cellular user and has no effect on the D2D communication. Simultaneously serving more cellular users in each cell will deteriorate both  SE and EE of a cellular user and D2D transmitter.

\item Simulation results validate our analysis and demonstrate the effectiveness of the proposed power control solution. {Our results show that when the D2D communication distance moderately increases, the SE and EE of a cellular user is comparable to that of a D2D transmitter.}

\end{itemize}

The remainder of this paper is organized as follows. Section II presents the proposed system model and the power control mechanism. Section III evaluates the SE and EE of the cellular and D2D links. Numerical results are provided in Section IV and
conclusion is drawn in Section V.

\section{System Description}
As shown in Fig. 1,  we consider uplink transmission in a cellular network, where massive MIMO enabled macrocells are underlaid with D2D transceivers,  i.e., they share the same frequency bands. The locations of macrocell base stations (MBSs) are modeled following a homogeneous Poisson point process (HPPP) $\Phi_\mathrm{M}$ with
density $\lambda_\mathrm{M}$. The locations of D2D transmitters are modeled following an independent
HPPP $\Phi_\mathrm{D}$ with density $\lambda_\mathrm{D}$. Each MBS is equipped with $N$ antennas and receives data
streams from $S$ single-antenna cellular user equipments (CUEs) over the same time and frequency band,
while each D2D receiver equipped with one single antenna receives one data stream from a single-antenna D2D transmitter in each transmission.
The linear zero-forcing beamforming  (ZFBF) is employed to cancel the intra-cell interference at the MBS~\cite{Hosseini2014_Massive}. It is assumed that the
density of CUEs is much greater than that of MBSs so that there always will be multiple active CUEs in every macrocell. Each channel
undergoes independent and identically distributed (i.i.d.) quasi-static Rayleigh fading. {Each CUE is assumed to be connected with
its nearest MBS such that the Euclidean plane is divided into Poisson-Voronoi cells~\cite{Kaibin2014,Emil_2016_small_cell_massive_MIMO}.}
\begin{figure}\label{smd2d}
\centering
\includegraphics[width=3.2 in, height=2.7 in]{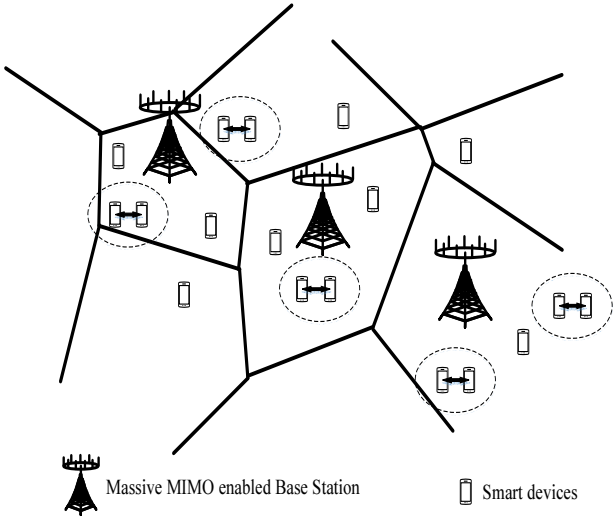}
\caption{An illustration of the D2D underlaid cellular networks equipped with massive MIMO MBSs.}
\end{figure}


\subsection{Power Control Policy}
In the macrocells, the open-loop uplink power control is applied such that far-away CUEs can obtain more path loss compensation}, and the transmit power for a CUE associated with the MBS is given by\footnote{Note that \cite{Novlan_2013_TWC} also studied the open-loop power control in a single-tier cellular networks without considering the maximum transmit power constraint.}
\begin{align}\label{power_control_MBS}
{P_{{\mathrm{C}}}}=\min\left\{P_\mathrm{max}^{\mathrm{C}},P_o \left( L\left(\left|X_{\mathrm{C},\mathrm{M}}\right|\right)\right)^{ -\eta} \right\},
\end{align}
where $P_\mathrm{max}^{\mathrm{C}}$ is the maximum transmit power, $P_o$ is the normalized power density, $L\left(\left|X_{\mathrm{C},\mathrm{M}}\right|\right)=\beta{ \left({{\left|X_{\mathrm{C},\mathrm{M}}\right|}}\right)^{-\alpha_\mathrm{M}}}$, $\alpha_\mathrm{M}$ is the path loss exponent, $\beta$ is the frequency dependent constant value, $\left|X_{\mathrm{C},\mathrm{M}}\right|$ is the distance between the CUE and its associated MBS and $\eta \in [0,1]$ is the path loss compensation factor, which controls the CUE transmit power. Here $\eta = 1$ represents
that the path loss between a CUE and its serving MBS is fully
compensated, and $\eta = 0$ represents that there is no path loss
compensation.  Note that the open-loop uplink power control does not require the instantaneous CSI.

To mitigate the D2D-to-cellular interference, we consider that the average received interference at the MBS from a D2D transmitter should not exceed a maximum value $I_\mathrm{th}$ under maximum D2D transmit power constraint, which is different from \cite{K_S_Ali_2015} where D2D transmitter stops transmissions if its transmit power is larger than a predefined value to achieve a fixed data rate. Therefore, the D2D transmit power is given by
\begin{align} \label{Power_control_D2D}
{P_{{\mathrm{D}}}}=\min\left\{P_\mathrm{max}^\mathrm{D},\frac{I_\mathrm{th}}{L\left(\left|X_{\mathrm{D}, \mathrm{M}}\right|\right)}\right\},
\end{align}
where $P_\mathrm{max}^\mathrm{D}$ is the maximum D2D transmit power, $\left|X_{\mathrm{D}, \mathrm{M}}\right|$ is the distance between a D2D transmitter and its nearest MBS. If there is no power control on the D2D transmitters, the shorter $\left|X_{\mathrm{D}, \mathrm{M}}\right|$, the stronger interference power. Here, $I_\mathrm{th}=0$ represents that there is no allowable D2D transmission and the considered network reduces to the massive MIMO enabled multi-cell network, and $I_\mathrm{th}=\infty$ represents that there is no D2D power control. {Different from \cite{lin2015interplay} which assumed that MBSs can obtain the instantaneous CSI between the D2D transmitters and themselves, and have the ability of canceling sufficient number of D2D interfering signals, the proposed D2D power control policy does not need  the instantaneous CSI and possesses much lower complexity.}

\subsection{Channel Model}
We assume that a typical serving MBS is located at the origin $o$. The receive signal-to-interference-plus-noise ratio (SINR) of a typical serving MBS at a random distance $\left|X_{o,{\mathrm{M}}}\right|$ from its intended CUE is given by
\begin{align}\label{SINR_Macro}
\mathrm{SINR}_\mathrm{M}= \frac{{{{P_{o,\mathrm{C}}}}{h_{o,{\mathrm{M}}}} L\left( {\left|{X_{o,{\mathrm{M}}}}\right|} \right)}}{{I_\mathrm{M}+I_\mathrm{D}+ {\sigma^2}}},
\end{align}
where ${P_{o,\mathrm{C}}}$ is the transmit power of the typical CUE, $h_{o,{\mathrm{M}}}\sim \Gamma\left(N-S+1,1\right)$~\cite{Hosseini2014_Massive} is the small-scale fading channel power gain between the typical serving MBS and its intended CUE { \footnote{ {Since the cellular and D2D users may experience similar shadow fading conditions which are not independent, to be tractable, the effect of shadow fading is not examined in this paper.}}}, $\sigma^2$ is the noise power, $I_\mathrm{M}$ and $I_\mathrm{D}$ are the interference from inter-cell CUEs and D2D transmitters, which are found as
\begin{equation}
\left\{\begin{aligned}
I_\mathrm{M}&=\sum\nolimits_{i \in {\Phi_{\mathrm{u},\mathrm{M}}}\backslash \mathrm{B}\left(o\right)} {{P_{i,\mathrm{C}}}
{h_{i,\mathrm{M}}}L\left( {\left|X_{i,\mathrm{M}}\right|} \right)},\\
I_\mathrm{D}&=\sum\nolimits_{j \in {\Phi_\mathrm{D}}} {{P_{j,{\mathrm{D}}}}
{h_{j,\mathrm{M}}}L\left( \left|X_{j,\mathrm{M}}\right| \right)},
\end{aligned}\right.
\end{equation}
where  ${P_{i,\mathrm{C}}}$ is the transmit power of the interfering CUE $i \in {\Phi_{\mathrm{u},\mathrm{M}}}\backslash \mathrm{B}\left(o\right)$ ($\Phi_{\mathrm{u},\mathrm{M}} \backslash \mathrm{B}\left(o\right)$ is the point process corresponding to the interfering CUEs), ${P_{j,{\mathrm{D}}}}$ is the transmit power of the interfering D2D $j\in {\Phi_\mathrm{D}}$,  ${h_{i,\mathrm{M}}}\sim \rm{exp}(1)$ and $\left|X_{i,\mathrm{M}}\right|$ are the small-scale fading interfering channel power gain and
distance between the typical serving MBS and interfering CUE $i$, respectively, ${h_{j,\mathrm{M}}}\sim \rm{exp}(1)$ and $\left|X_{j,\mathrm{M}}\right|$ are the small-scale fading interfering channel power gain and distance between the typical serving MBS and interfering D2D transmitter $j$, respectively.

{We adopt the dipole model in which each D2D transmitter has a corresponding receiver at distance $d_o$~\cite{lee2015power,lin2015interplay,Baccelli2009}, since D2D network has  ad hoc functionality~\cite{Zhu_Kit_Robert_2017}.} The SINR of a typical D2D receiver   from its D2D transmitter is given by
\begin{align}\label{SINR_Small}
\mathrm{SINR}_\mathrm{D}= \frac{{{P_{o,\mathrm{D}}}{g_{o,\mathrm{D}}}L\left( d_o \right)}}{{J_\mathrm{M}+J_\mathrm{D}+ {\sigma^2}}},
\end{align}
where { $g_{o,\mathrm{D}}\sim \mathrm{exp}(1)$ and $L\left( d_o \right)=\beta{ \left(d_o\right)^{-\alpha_\mathrm{D}}}$ is the small-scale fading channel power gain and path loss between the typical  D2D receiver and its corresponding D2D transmitter, respectively, $\alpha_\mathrm{D}$ is the path loss exponent,} $J_\mathrm{M}$ and $J_\mathrm{D}$ are the interference from the CUEs and interfering D2D transmitters, respectively, given by
\begin{equation}
\left\{\begin{aligned}
J_\mathrm{M}&=\sum\nolimits_{i \in {\Phi_{\mathrm{u},\mathrm{M}}}} {{P_{i,\mathrm{C}}}
{g_{i,\mathrm{D}}}L\left( {\left|X_{i,\mathrm{D}}\right|} \right)},\\
J_\mathrm{D}&=\sum\nolimits_{j \in {\Phi_\mathrm{D} \backslash o}} {{P_{{j},\mathrm{D}}}
{g_{j,\mathrm{D}}}L\left( {\left|X_{j,\mathrm{D}}\right|} \right)},
\end{aligned}\right.
\end{equation}
where  ${g_{i,\mathrm{D}}}\sim \rm{exp}(1)$ and $\left|X_{i,\mathrm{D}}\right|$ are the small-scale fading interfering channel power gain and
distance between the typical  D2D receiver and interfering CUE $i\in {\Phi_{\mathrm{u},\mathrm{M}}}$, respectively, ${g_{j,\mathrm{D}}}\sim \rm{exp}(1)$ and $\left|X_{j,\mathrm{D}}\right|$ are the small-scale fading interfering channel power gain and
distance between the typical D2D receiver and interfering D2D transmitter $j\in {\Phi_\mathrm{D}\backslash o}$, respectively.

\section{Spectral and Energy Efficiency}
By addressing the effects of power control, we examine the SE and EE for the cellular and D2D transmissions. We first need to derive the following probability density function (PDF) of the D2D transmit power based on \eqref{Power_control_D2D}.

\subsection{D2D Transmit Power Distribution}
\begin{lemma}\label{Lemma1}
The PDF of a typical D2D transmit power is given by
\setcounter{equation}{6}\begin{align}\label{PDF_D2D}
f_{{P_{{\mathrm{D}}}}} \left(x\right)=\left\{ \begin{array}{l}
\frac{{2\pi \lambda_\mathrm{M} }}{{{\alpha _{\rm{M}}}}}{\left( \frac{\beta }{I_{{\rm{th}}}} \right)^{2/{\alpha _{\rm{M}}}}}{x^{ 2/{\alpha _{\rm{M}}} - 1}}\bar{\Delta}\left(x\right),\;\,x < P_\mathrm{max}^\mathrm{D} \\
\delta\left( {x - P_\mathrm{max}^\mathrm{D} } \right) \bar{\Delta}\left(P_\mathrm{max}^\mathrm{D}\right),\;\,x \ge P_\mathrm{max}^\mathrm{D}
\end{array} \right.,
\end{align}
where $\bar{\Delta}\left(x\right)=\exp\left(-\pi \lambda_\mathrm{M} {\left( {\frac{\beta x}{I_{{\rm{th}}}}} \right)}^{2/{\alpha _{\rm{M}}}}\right)$ and $\delta\left(\cdot\right)$ is the Dirac delta function.
\end{lemma}

\begin{proof}
The proof is provided in Appendix A.

\end{proof}
From \textbf{Lemma} \ref{Lemma1}, we see that the level of the D2D transmit power is dependent on the massive MIMO enabled MBS density and the interference threshold $I_\mathrm{th}$.

\subsection{Spectral Efficiency}
{With the assistance of \textbf{Lemma} \ref{Lemma1}, the SE for a typical CUE and D2D transmitter can be derived in the following \textbf{Theorem 1} and \textbf{Theorem 2}, respectively. Note that the results in \textbf{Theorem 1} and \textbf{Theorem 2} are general and include the special case of fixed transmit power.
\begin{theorem}
\label{cuerate}
The  SE under power control for a typical CUE is given by
\begin{align}\label{macrocell_capacity}
\overline{R}_\mathrm{C}=1/\ln(2) \int_0^\infty \frac{\Xi_1\left(t\right)}{t} \Xi_2\left(t\right) e^{-\sigma^2 t} dt,
\end{align}
where $\Xi_1\left(t\right)$ and $\Xi_2\left(t\right)$ is given by \eqref{Xi_1} and \eqref{Xi_2} at the top of next page, in which $r_o=\left( {\frac{{{P_{{\rm{max}}}^\mathrm{C}}}}{{{P_o}}}} \right)^{1/({\alpha_\mathrm{M}}\eta) }{\beta^{1/{\alpha_\mathrm{M}}}}$  and $\varpi_0=\pi \lambda_\mathrm{M}{\left( \frac{\beta P_\mathrm{max}^\mathrm{D}}{I_{{\rm{th}}}} \right)^{2/{\alpha _{\rm{M}}}}}$.

\begin{figure*}[!t]
\normalsize
\begin{align}\label{Xi_1}
&\Xi_1\left(t\right)=\int_{\rm{0}}^\infty \left( {1- e^{ - t{{{P_{o,\mathrm{C}}}}(N-S+1) \beta (\frac{x}{\pi \lambda_\mathrm{M}})^{-\frac{\alpha_\mathrm{M}}{2}}} }} \right) \exp\Big(-2\pi S \lambda_\mathrm{M} \int_x^\infty (1-\Upsilon_1)  rdr  -x\Big) dx
\end{align}
with
\begin{align*}
\Upsilon_1=\left(1+t  {P_\mathrm{max}^{\mathrm{C}}}
\beta r^{-\alpha_\mathrm{M}}\right)^{-1} e^{-\pi \lambda_\mathrm{M} r_o^2} +\int_0^{\pi \lambda_\mathrm{M} r_o^2} \left(e^{\nu}+e^{\nu} t P_o \beta^{1-\eta} (\frac{\nu}{\pi \lambda_\mathrm{M}})^{\frac{\eta \alpha_\mathrm{M}}{2}}
 r^{-\alpha_\mathrm{M}}\right)^{-1}  d\nu
\end{align*}
\hrulefill
\begin{align}\label{Xi_2}
&\Xi_2\left(t\right)=\exp\Big\{-\pi \lambda_\mathrm{D} \beta^{\frac{2}{\alpha_\mathrm{M}}} \Big[\frac{{(P_\mathrm{max}^\mathrm{D})^{\frac{2}{{\alpha _{\rm{M}}}}}}}{\varpi_0}(1-e^{-\varpi_0}-\varpi_0 e^{-\varpi_0})+ \left(P_\mathrm{max}^\mathrm{D}\right)^{\frac{2}{\alpha_\mathrm{M}}} {\bar{\Delta}\left(P_\mathrm{max}^\mathrm{D}\right)}\Big] \nonumber \\
& \Gamma(1+\frac{2}{\alpha_\mathrm{M}}) \Gamma(1-\frac{2}{\alpha_\mathrm{M}}) t^\frac{2}{\alpha_\mathrm{M}}  \Big\}
\end{align}
\hrulefill
\end{figure*}

\end{theorem}
\begin{proof}
The proof is provided in Appendix B.
\end{proof}  }

{It is indicated from \textbf{Theorem 1}  that the SE of a typical CUE $\overline{R}_\mathrm{C}$ is an increasing function of $N$, since add more massive MIMO antennas will increase power gains. It is a decreasing function of $S$, since serving more CUEs in each cell will decrease the power gain and increase the inter-cell interference. In addition, $\overline{R}_\mathrm{C}$ is also a decreasing function of $\lambda_\mathrm{D}$, since more D2D transmissions will give rise to severer D2D-to-cellular interference. In addition, when $\eta=0$ and $I_\mathrm{th}=\infty$, the result given in \eqref{macrocell_capacity} reduces to the fixed transmit power case.}

Based on \textbf{Theorem} \ref{cuerate}, the area SE (bps/Hz/m$^{2}$) achieved by the cellular is calculated as
\begin{align}\label{Area_rate_cellular}
\mathcal{A}_\mathrm{C}= \overline{R}_\mathrm{C} S  \lambda_\mathrm{M}.
\end{align}

 D2D density plays a dominant role in the level of D2D-to-cellular interference, which has a big effect on the cellular SE. Thus, we have the following important scale property.
\begin{Scale_property}
{Given a targeted SE} $\overline{R}_\mathrm{C}^\mathrm{th}$ of the CUE, it is achievable when the D2D density satisfies
{\begin{align}\label{corollary_1}
\lambda_\mathrm{D} \leq \left((N-S+1) \frac{X_1}{2^{\overline{R}_\mathrm{C}^\mathrm{th}}-1}-X_2\right)(X_3)^{-1},
\end{align}}
where
{ \begin{align}\label{X_1_theo}
X_1 & = \exp\Bigg\{\ln \left( {\frac{{{P^\mathrm{C}_{\max }}}}{{{P_o}{\beta ^{ - \eta }}}}} \right)\exp \left( { - \pi {\lambda _\mathrm{M}} r_o^2} \right) + \ln \left( {{P_o}{\beta ^{ - \eta + 1}}} \right) -  \nonumber \\
& \frac{\alpha_\mathrm{M}}{2}\psi\left( 1 \right) +\frac{\alpha_\mathrm{M}}{2}\ln \left( {\pi {\lambda_\mathrm{M}}} \right) + \frac{\eta {\alpha_\mathrm{M}} }{2}\big(\ell + \nonumber\\
&\Gamma\left(0,r_o^2 \pi {\lambda_\mathrm{M}}\right)+2e^{-r_o^2 \pi {\lambda_\mathrm{M}}}\ln(r_o)+\ln(\pi {\lambda_\mathrm{M}})\big) \Bigg\}
\end{align} }
with the digamma function $\psi\left(\cdot\right)$~\cite{Abramowitz} and Euler-Mascheroni constant $\ell \approx 0.5772$, and
\begin{equation}
\hspace{-0.5 cm}\left\{\begin{aligned}
X_2 &=2 \pi \lambda_\mathrm{M} \int_0^\infty \beta 2\pi S{\lambda_\mathrm{M}} \bar{P}_\mathrm{C} \frac{x^{2- {\alpha_\mathrm{M}}}}{\alpha_\mathrm{M}-2}  x \exp\left(-\pi \lambda_\mathrm{M}  x^2\right) dx +\sigma^2,\\
X_3&=2 \pi \beta   \left(\frac{D_o^{2 - {\alpha_\mathrm{M}}}}{2} +\frac{D_o^{2 - {\alpha_\mathrm{M}}}}{{\alpha_\mathrm{M}}-2}\right)\bar{P}_\mathrm{D},
\end{aligned}\right.
\end{equation}
where $\bar{P}_\mathrm{C}$ and $\bar{P}_\mathrm{D}$ given by \eqref{Lambda_1} and \eqref{Lambda_2} (at the top of next page), and represent the average transmit powers of CUE and D2D transmitter, respectively,  ${\rm{\mathbf{1}}}\left(A\right)$ is the indicator function that returns one if the condition $A$ is satisfied,  $\varpi_1(x)=\max\left\{D_o,{{\left( {\frac{{\beta x}}{I_\mathrm{th}}} \right)}^{1/{\alpha _{\rm{M}}}}}\right\}$, and $D_o$ is the reference distance, which is utilized to avoid singularity caused by proximity~\cite{Kaibin2014}\footnote{Note that the reference distance can also represent the minimum distance between a D2D transmitter and the typical serving MBS in the practical scenario~\cite{Heath_ITA_2012}.}.
\begin{figure*}[!t]
\normalsize
\begin{align}\label{Lambda_1}
\bar{P}_\mathrm{C}=P_o \beta^{-\eta} {\left( {\pi {\lambda_\mathrm{M}}} \right)^{ - \frac{{\eta {\alpha_\mathrm{M}}}}{2}}}\left( {\Gamma \left( {1 + \frac{{\eta {\alpha_\mathrm{M}}}}{2}} \right) - \Gamma \left( {1 + \frac{{\eta {\alpha _\mathrm{M}}}}{2},\pi {\lambda_\mathrm{M}}\sqrt {{r_o}} } \right)} \right) + {{P^\mathrm{C}_{\max }}} \exp \left( { - \pi {\lambda_\mathrm{M}} r_o^2} \right)
\end{align}
\hrulefill
\begin{align}
\label{Lambda_2}
\bar{P}_\mathrm{D}=\int_0^{P_\mathrm{max}^\mathrm{D}} \left({\rm{\mathbf{1}}}\left( x < \frac{D_o^{\alpha _{\rm{M}}}I_\mathrm{th}}{\beta}  \right) \left(1-\exp\left(-\pi \lambda_\mathrm{M} D_o^2\right) \right)+\exp\left(-\pi \lambda_\mathrm{M} (\varpi_1(x))^2\right)\right) dx
\end{align}
\hrulefill
\end{figure*}
\end{Scale_property}
\begin{proof}
The proof is provided in Appendix C.
\end{proof}

{ From \textbf{Scale Property 1}, we find that given a targeted SE, the number of D2D links needs to be lower than a critical value, to limit the D2D-to-cellular interference. Adding more massive MIMO antennas can allow cellular networks to accommodate more underlaid  D2D links.}

For a typical D2D link, its SE can be obtained as follows.
\begin{theorem}
The SE for a typical D2D link with a given distance $d_o$ is given by
\begin{align}\label{d2d_capacity}
\overline{R}_\mathrm{D}=1/\ln(2) \int_0^\infty  \frac{1}{t}\left(1 - \Xi_3\left(t\right)\right) \Xi_4\left(t\right){e^{ - t{\sigma ^2}}}dt,
\end{align}
where $\Xi_3\left(t\right)$ and $\Xi_4\left(t\right)$ are given by \eqref{Xi_3} and \eqref{Xi_4} at the next page.
\end{theorem}
\begin{figure*}[!t]
\normalsize
\begin{align}\label{Xi_3}
&\Xi_3\left(t\right)=\int_0^{\varpi_0} \frac{e^{-x}}{1+t{I_\mathrm{th}}(\frac{x}{\pi \lambda_\mathrm{M}})^{\frac{\alpha_\mathrm{M}}{2}}  d_o^{-\alpha_\mathrm{D}}}   dx
+ \frac{\bar{\Delta}\left(P_\mathrm{max}^\mathrm{D}\right)}{1+t {P_\mathrm{max}^\mathrm{D}} \beta d_o^{-\alpha_\mathrm{D}}}
\end{align}
\hrulefill
\begin{align}\label{Xi_4}
&\Xi_4\left(t\right)=\exp\Big(-\pi \beta^{\frac{2}{\alpha_\mathrm{D}}} \Gamma(1+\frac{2}{\alpha_\mathrm{D}}) \Gamma(1-\frac{2}{\alpha_\mathrm{D}}) t^\frac{2}{\alpha_\mathrm{D}}(S \lambda_\mathrm{M} \Omega_1 + \lambda_\mathrm{D} \Omega_2)  \Big),
\end{align}
where
\begin{align}\label{Omega_123}
&\Omega_1=(P_o \beta^{-\eta})^{\frac{2}{\alpha_\mathrm{D}}} (\pi \lambda_\mathrm{M})^{-\frac{\eta \alpha_\mathrm{M}}{\alpha_\mathrm{D}}}
\Big(\Gamma(1+\frac{\eta \alpha_\mathrm{M}}{\alpha_\mathrm{D}})-\Gamma(1+\frac{\eta \alpha_\mathrm{M}}{\alpha_\mathrm{D}},\pi \lambda_\mathrm{M} r_o^2)\Big)+\big(P_\mathrm{max}^{\mathrm{C}}\big)^{\frac{2}{\alpha_\mathrm{D}}} e^{-\pi \lambda_\mathrm{M} r_o^2}
\end{align}
\begin{align}\label{Omega_223}
&\Omega_2=\left(\pi \lambda_\mathrm{M}( \frac{\beta }{I_{{\rm{th}}}})^{\frac{2}{\alpha _{\rm{M}}}}\right)^{-\frac{\alpha _{\rm{M}}} {\alpha _{\rm{D}}}} \Big(\Gamma(1+\frac{\alpha_\mathrm{M}}{\alpha_\mathrm{D}})-\Gamma(1+\frac{\alpha_\mathrm{M}}{\alpha_\mathrm{D}},\varpi_0)\Big)+ \left(P_\mathrm{max}^\mathrm{D}\right)^{\frac{2}{\alpha_\mathrm{D}}} {\bar{\Delta}\left(P_\mathrm{max}^\mathrm{D}\right)}
\end{align}

\hrulefill
\end{figure*}

\begin{proof}
The proof is provided in Appendix D.
\end{proof}

{It is indicated from \textbf{Theorem 2}  that the SE for a typical D2D link is independent of massive MIMO antennas, and is a decreasing function of $S$ due to the fact that more uplink transmissions will result in severer cellular-to-D2D interference. Moreover, it is also a decreasing function of $\lambda_\mathrm{D}$, since more inter-D2D interference deteriorates the typical D2D transmission.}

Based on \textbf{Theorem} 2, the area SE achieved by the D2D tier is
\begin{align}\label{Area_rate_D2D_1}
\mathcal{A}_\mathrm{D}= \overline{R}_\mathrm{D}   \lambda_\mathrm{D}.
\end{align}
%
%
%
Since D2D density also has a substantial effect on the level of inter-D2D interference, which greatly affects the SE of D2D. Thus, we have the following important scale property.
\begin{Scale_property}\label{corollary_11}
The {targeted SE} $\overline{R}_\mathrm{D}^\mathrm{th}$ of the D2D transmitter can be achieved when the D2D density satisfies
{\begin{align}\label{corollary_2}
\lambda_\mathrm{D} \leq \left(\frac{X_4}{2^{\overline{R}_\mathrm{D}^\mathrm{th}}-1}-S {\lambda _\mathrm{M}} X_5\right)(X_6)^{-1},
\end{align} }
where
{\begin{align}
\label{X_4_theo}
{X_4} = &\exp \bigg\{ \int_0^{P_{\max }^{\rm{D}}} {\ln \left( x \right){f_{{P_{\rm{D}}}}}\left( x \right)dx + \ln \left( {P_{\max }^{\rm{D}}} \right)}  \nonumber\\
&\qquad \qquad\bar \Delta \left( {P_{\max }^{\rm{D}}} \right)+ \ln \left( {\beta {d_o}^{ - {\alpha _{\rm{D}}}}} \right) + \ell \bigg\},\\
 {X_5} =&2\pi \bigg(P_o \beta^{-\eta} {\left( {\pi {\lambda_\mathrm{M}}} \right)^{ - \frac{{\eta {\alpha_\mathrm{M}}}}{2}}}\Big( \Gamma \left( {1 + \frac{{\eta {\alpha_\mathrm{M}}}}{2}} \right)  \nonumber\\
&- \Gamma \left( {1 + \frac{{\eta {\alpha _M}}}{2},\pi {\lambda_\mathrm{M}}\sqrt {{r_o}} } \right) \Big)+ {{P^\mathrm{C}_{\max }}} \nonumber\\
&\exp \left( { - \pi {\lambda_\mathrm{M}} r_o^2} \right)\bigg) \beta \left(\frac{D_1^{2 - {\alpha_\mathrm{D}}}}{2} +\frac{D_1^{2 - {\alpha_\mathrm{M}}}}{{\alpha_\mathrm{D}}-2}\right),\\
{X_6} =& 2\pi \beta  \left(\frac{D_2^{2 - {\alpha_\mathrm{D}}}}{2} +\frac{D_2^{2 - {\alpha_\mathrm{D}}}}{{\alpha_\mathrm{D}}-2}\right)\bar{P}_\mathrm{D},
\end{align}}
in which $D_1$  and $D_2$ are the reference distances, $\bar{P}_\mathrm{D}$ is given by \eqref{Lambda_2}.

\end{Scale_property}
\begin{proof}
The proof is provided in Appendix E.
\end{proof}

{ From \textbf{Scale Property 2}, we find that given a targeted SE, the number of D2D links needs to be lower than a critical value, to limit the inter-D2D interference. The number of D2D links that achieves the targeted SE decreases when each MBS serves more users at the same time and frequency band, due to severer cellular-to-D2D interference.}

\subsection{Energy Efficiency}
In this subsection, { we evaluate the EE of cellular and D2D transmissions, which is of paramount importance in 5G systems due to the fact that one of key performance indicators (KPIs) in 5G is ten times lower energy consumption per service than the today's networks~\cite{5GPPP_2015}. In this paper, one of our aims is to find out whether the uplink EE of massive MIMO cellular networks is comparable to that of D2D.}  {The EE is defined as the ratio of the SE to the average power consumption.}

The average power consumption of a CUE is calculated as
\begin{equation}
\label{Averag_2016}
\overline {P}^{total}_\mathrm{C}=P_f+\frac{\mathrm{\overline{P}_\mathrm{C}}}{\zeta},
\end{equation}
where $P_f$ is the fixed circuit power consumption,  $\mathrm{\overline{P}_\mathrm{C}}$ is the average transmit power given by \eqref{Lambda_1}, and $\zeta$ is the power amplifier efficiency. Thus, the EE for a typical CUE is derived as
\begin{align}
\label{CUE_U_EE}
\mathrm{\overline{EE}_C}=\frac{\overline{R}_\mathrm{C}}{ \overline {P}^{total}_\mathrm{C}},
\end{align}
where $\overline{R}_\mathrm{C}$ is the average SE given by \eqref{macrocell_capacity}. {{For uplink transmission, the average power consumption for a CUE is only dependent on the maximum transmit power level and the path loss compensation, as shown in \eqref{Averag_2016}. Therefore, $\mathrm{\overline{EE}_C}$ is an increasing function of $N$ and a decreasing function of $S$, since $\overline{R}_\mathrm{C}$ increases with $N$ and decreases with increasing $S$, according to \textbf{Theorem 1}. }}

Likewise,  the average power consumption of a D2D transmitter is calculated as
\begin{equation}
\overline {P}^{total}_\mathrm{D}=P_f+\frac{\mathrm{\overline{P}_\mathrm{D}}}{\zeta},
\end{equation}
where $\mathrm{\overline{P}_\mathrm{D}}$ is the average transmit power. Based on \eqref{power_D2D_exp} in Appendix B, $\mathrm{\overline{P}_\mathrm{D}}$ is given by
\begin{align}\label{average_P_D}
\mathrm{\overline{P}_\mathrm{D}}=\frac{{(P_\mathrm{max}^\mathrm{D})^{{2}}}}{\varpi_2}(1-e^{-\varpi_2}-\varpi_2 e^{-\varpi_2})+ \left(P_\mathrm{max}^\mathrm{D}\right)^{{2}} {\bar{\Delta}\left(P_\mathrm{max}^\mathrm{D}\right)}
\end{align}
with $\varpi_2=\pi \lambda_\mathrm{M}\left( \frac{\beta P_\mathrm{max}^\mathrm{D}}{I_{{\rm{th}}}} \right)^{2}$. Thus, the EE for a typical D2D pair is derived as
\begin{align}
\label{D2D_U_EE}
\mathrm{\overline{EE}_D}=\frac{\overline{R}_\mathrm{D}}{\overline {P}^{total}_\mathrm{D}},
\end{align}
where $\overline{R}_\mathrm{D}$ is the average SE given by \eqref{d2d_capacity}. {{Similarly, the EE for a typical D2D pair is independent of massive MIMO antennas, and  is a decreasing function of $S$, since more cellular-to-D2D interference decreases the SE.}}

\begin{table*}[tb]
\centering
\caption{Simulation Parameters}\label{table1}
\begin{tabular}{ |p{6cm}|p{3.5cm}| p{4.2cm}|}
 \hline
 Parameter & Symbol & Value \\
 \hline
 Pathloss exponent to MBS & $\alpha_\mathrm{M}$ & 3.5\\
  \hline
 Pathloss exponent to D2D &  $\alpha_\mathrm{D}$  & 4  \\
  \hline
 The maximum transmit power of MUE & $P^\mathrm{C}_\mathrm{max}$ & 23 dBm\\
  \hline
 Bandwidth   & BW  & 5 MHz \\
  \hline
 The noise power & $\sigma^2$ & $-170+10 \times \log_{10}(\mathrm{BW})$ dBm\\
  \hline
 The power density~\cite{safjan2013open} & $P_o$ & -80 dBm\\
  \hline
Static power consumption & $P_f$ & 100 mW\\
 \hline
 Power amplifier efficiency & $\zeta$ & 0.5  \\
  \hline
\end{tabular}
\end{table*}

\section{Numerical Results}

In this section, numerical results are presented to evaluate the area average SE and average EE of the cellular and D2D in the D2D underlaid massive MIMO cellular network. Such a network is assumed to operate at a carrier frequency of 1 GHz. Our results show the effect of massive MIMO in terms of user number $S$, the effect of D2D in terms of its density $\lambda_D$ and the effect of power control in terms of the compensation factor $\eta$ and interference threshold $I_\mathrm{th}$. {The basic parameters that are adopted in all the  simulations are summarized in Table 1}, and  it is assumed that the density of MBSs is $\lambda_\mathrm{M}=\left(500^2 \times \pi\right)^{-1} \mathrm{m}^{-2}$ in a circular region with radius $1\times 10^4$ m.

In the figures, the analytical area SE curves for the cellular and D2D are obtained from \eqref{Area_rate_cellular} and \eqref{Area_rate_D2D_1}, respectively, and the analytical EE curves for a CUE or D2D transmitter are obtained from \eqref{CUE_U_EE} and \eqref{D2D_U_EE}, respectively. Monte Carlo simulated values of the uplink spectrum efficiency marked by `o' are numerically obtained to validate the analysis.

\subsection{Power Control Effect}
In this subsection, we illustrate the effects of power control on the area SE and EE,  to demonstrate the effectiveness of  the proposed power control solution.

\begin{figure}
\centering
\includegraphics[width=3.2 in, height=2.7 in]{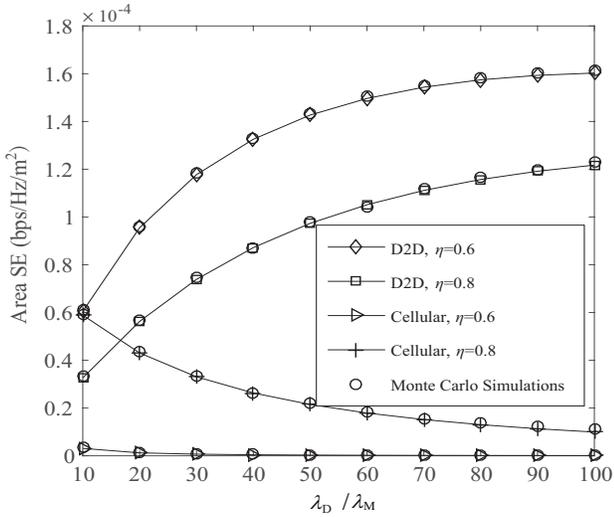}
\caption{ { Effects of D2D density with the variation of cellular power control on the area SE: $d_o=35$ m, $S=20$, $N=400$, $P^\mathrm{D}_\mathrm{{max}}= 15$  dBm and ${I_\mathrm{th}} / {\sigma^2} = 10$ dB.} }
\label{Fig1_CUE_PC}
\end{figure}
\begin{figure}
\centering
\includegraphics[width=3.7 in, height=2.8 in]{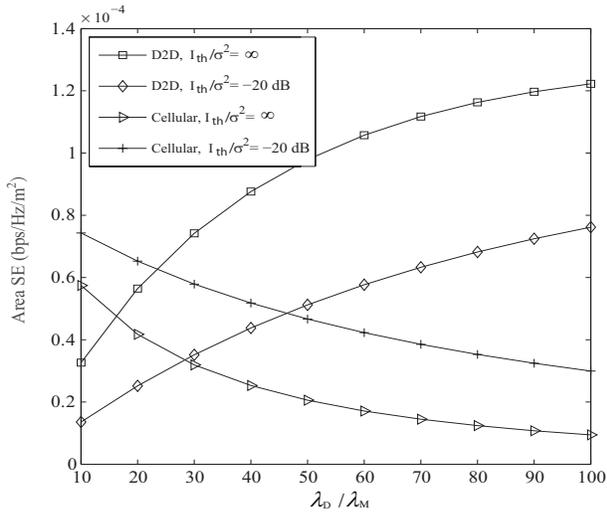}
\caption{  { Effects of D2D density with the variation of D2D power control on the area SE: $d_o=35$ m, $S=20$, $N=400$, $P^\mathrm{D}_\mathrm{{max}}= 15$  dBm and $\eta= 0.8$.}}
\label{Fig_I_th_SE}
\end{figure}

Fig. \ref{Fig1_CUE_PC} shows the effects of D2D density with the variation of cellular power control on the area SE. We see that uplink power control applied in the massive MIMO macrocells can significantly affect the area SE of the D2D and the cellular. Specifically, when the
transmit power of the CUE is controlled at a low level, the area SE of the D2D is improved, because D2D receivers experience less interference from the CUEs. In contrast, the area SE of the cellular decreases with the CUE transmit power. The  cellular performance is greatly degraded when the D2D links are dense, due to the severe interference from the D2D transmitters, which reveals that D2D-to-cellular interference mitigation is required for ensuring the uplink quality of service in the cellular networks.

Fig. \ref{Fig_I_th_SE} shows the effects of D2D density with the variation of D2D power control on the area SE. We observe that without D2D power control (i.e., ${I_\mathrm{th}} / {\sigma^2}=\infty$), the area SE of D2D tier is much higher than the massive MIMO aided cellular when D2D density is large. In particular, the area SE of the cellular is drastically deteriorated by the severe D2D-to-cellular interference. The implementation of the proposed D2D power control policy (e.g., ${I_\mathrm{th}} / {\sigma^2}=-20$ dB in this figure.) can efficiently mitigate the D2D-to-cellular interference, and thus improve the cellular performance.

\begin{figure}
\centering
\includegraphics[width=3.6 in, height=2.6 in]{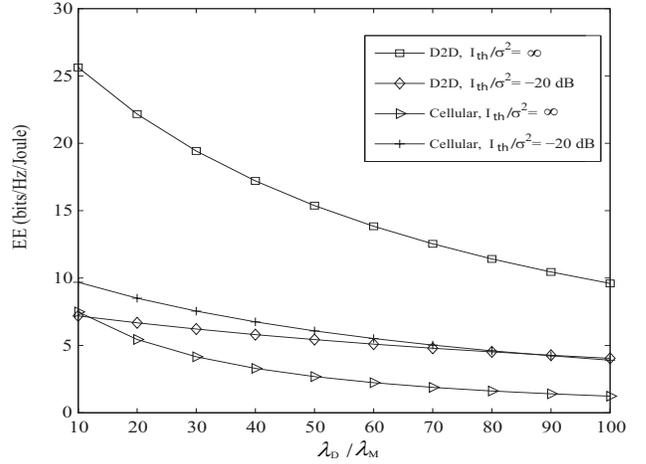}
\caption{ { Effects of D2D density with the variation of D2D power control on the EE: $d_o=35$ m, $S=20$, $N=400$, $P^\mathrm{D}_\mathrm{{max}}= 15$  dBm and $\eta= 0.8$.} }
\label{Fig_I_th_EE}
\end{figure}
\begin{figure}
\centering
\includegraphics[width=3.6 in, height=2.8 in]{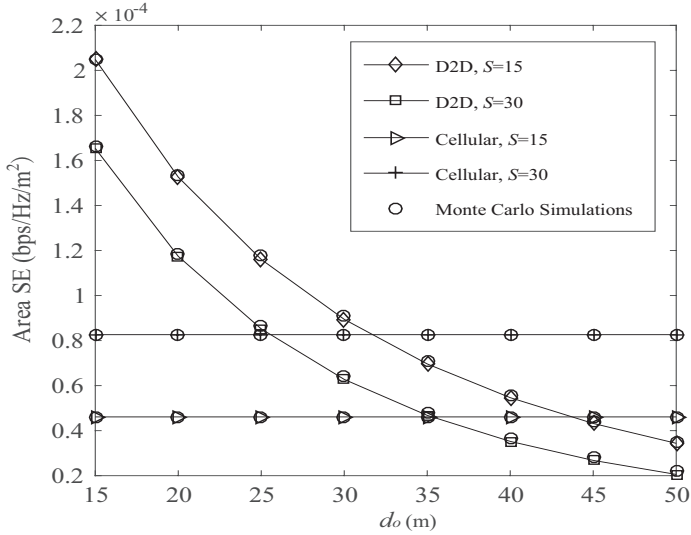}
\caption{ { Effects of D2D distance with the impact of massive MIMO on the area SE: $N=400$,  $\lambda_\mathrm{D}=30 \times \lambda_\mathrm{M}$, $P^\mathrm{D}_\mathrm{{max}}= 15$  dBm, $\eta= 0.9$ and
${I_\mathrm{th}} / {\sigma^2}= 5$ dB.}}
\label{Fig3_S}
\end{figure}

Fig. \ref{Fig_I_th_EE} shows the effects of D2D density with the variation of D2D power control on the EE. Without D2D power control, the EE of a D2D link is much higher than that of a cellular uplink, owing to the proximity. The interference increases with the D2D links, which harms both the EE of the cellular user and D2D user. The use of D2D power control enhances the EE of the cellular user, due to its SE improvement. Moreover, by properly coordinating the D2D-to-cellular interference,  the uplink EE of a massive MIMO aided cellular is comparable to that of a D2D link.

Fig. \ref{Fig3_S} shows the effects of D2D distance with the impact of massive MIMO on the area SE. It is obvious that when the distance between the D2D transmitter and its receiver grows large, the area SE of the D2D decreases, and it has no effect on the cellular performance. As more CUEs are served in each massive MIMO aided macrocell, there is a substantial increase in the area SE of the cellular, due to more multiplexing gains achieved by massive MIMO. However, when more CUEs are served in the uplink, the interference from CUEs is exacerbated, which degrades the D2D performance. Therefore, the cellular-to-D2D interference also needs to be coordinated. In addition, massive MIMO cellular can achieve better performance than D2D when the D2D distance is large.

\subsection{Massive MIMO Antennas Effect}
{In this subsection, we illustrate the effects of massive MIMO antennas on the area SE and EE. In the simulations, we set  $d_o=50$ m, $S=20$, $\lambda_\mathrm{D}=30 \times \lambda_\mathrm{M}$, $\eta=0.8$ and ${I_\mathrm{th}} / {\sigma^2} = 0$ dB.

\begin{figure}
\centering
\includegraphics[width=3.7 in, height=2.9 in]{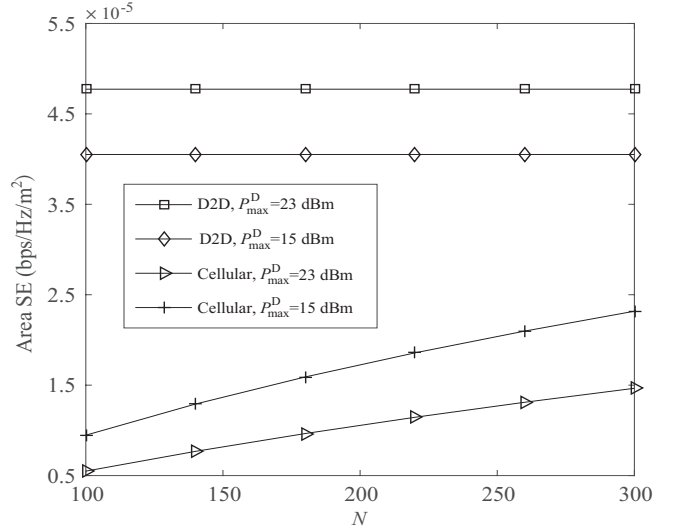}
\caption{ { Effects of massive MIMO antennas with the variation of maximum D2D transmit power on the area SE.}}
\label{Fig_massiveMIMO_antennas_SE}
\end{figure}

Fig.~\ref{Fig_massiveMIMO_antennas_SE} shows the effects of massive MIMO antennas with the variation of maximum D2D transmit power on the area SE. As confirmed in \textbf{Theorem 1}, the area SE increases with $N$ because of obtaining more power gains. As confirmed in \textbf{Theorem 2}, increasing massive MIMO antennas has no effect on the D2D SE. When larger maximum D2D transmit power is allowed, the area SE of the D2D is enhanced. However, the area SE of the cellular decreases due to the severer D2D-to-cellular interference.

\begin{figure}
\centering
\includegraphics[width=3.7 in, height=2.9 in]{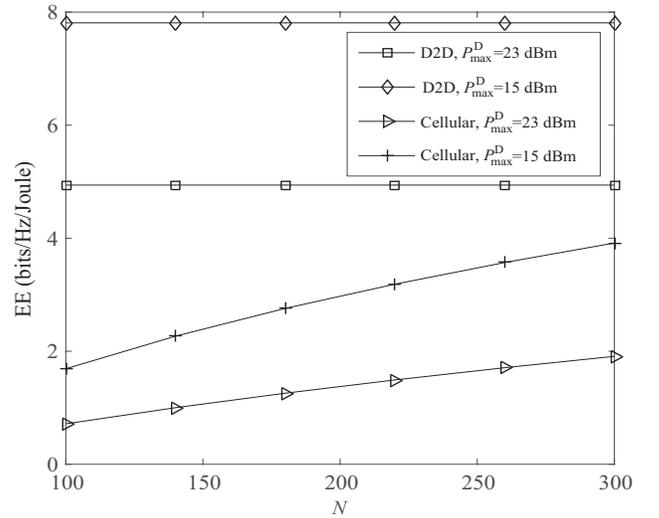}
\caption{ { Effects of massive MIMO antennas with the variation of maximum D2D transmit power on the EE.} }
\label{Fig_massiveMIMO_antennas_EE}
\end{figure}

Fig.~\ref{Fig_massiveMIMO_antennas_EE} shows the effects of massive MIMO antennas with the variation of maximum D2D transmit power on the EE.
As mentioned in Section III-C, the EE of a CUE increases with $N$ because of larger SE. Increasing $N$ has no effect on the EE of a D2D transmitter. Although Fig.~\ref{Fig_massiveMIMO_antennas_SE} shows that larger maximum D2D transmit power can improve the SE of the D2D, the EE of the D2D can be reduced because of more D2D power consumption. In addition, the EE of a CUE  decreases due to larger D2D-to-cellular interference.
}

\subsection{Interplay between Massive MIMO and D2D}

\begin{figure}[tb]
\centering
\subfigure[ Area SE of the cellular.] {\label{Fig_mue_se_3D}
\includegraphics[width=3.2 in, height=2.6 in]{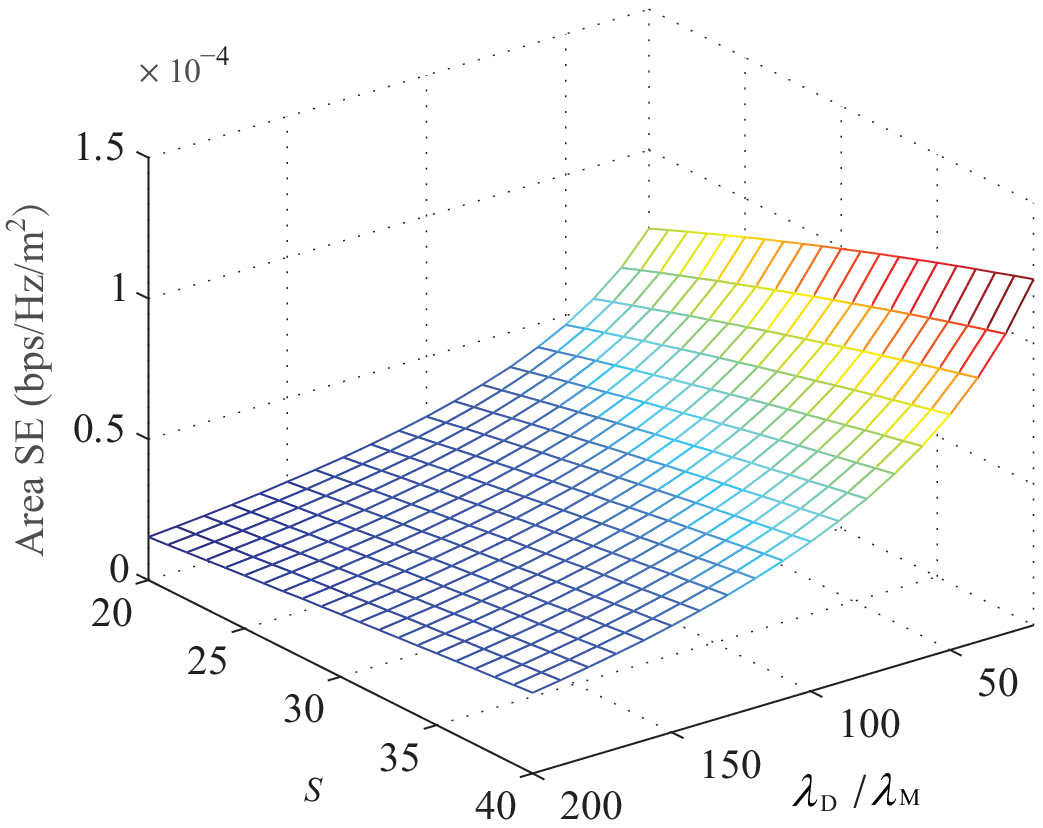}
}
\subfigure[Area SE of the D2D.]{
\label{Fig_d2d_se_3D}
  \includegraphics[width=3.2 in, height=2.6 in]{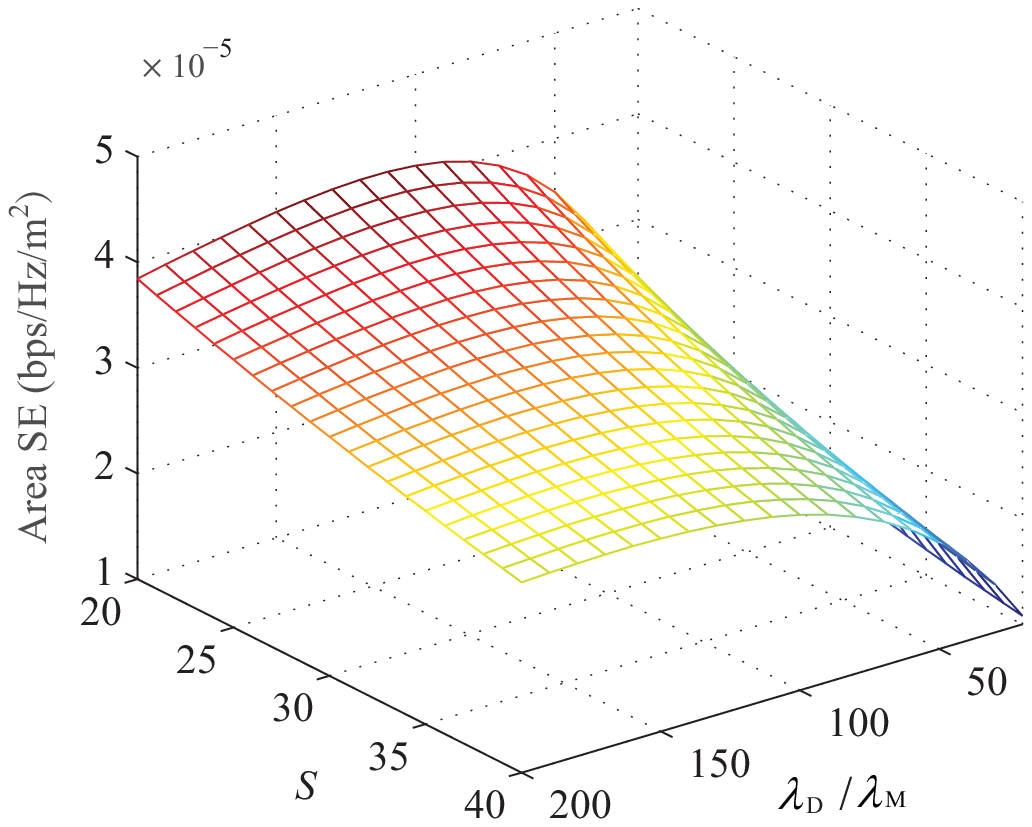}}
 \caption{Area SE of uplink D2D underlaid massive MIMO cellular networks.}
\label{Fig5}
\end{figure}

\begin{figure}[tb]
\centering
\subfigure[EE of a CUE.] {\label{Fig_mue_ee_3D}
\includegraphics[width=3.2 in, height=2.6 in]{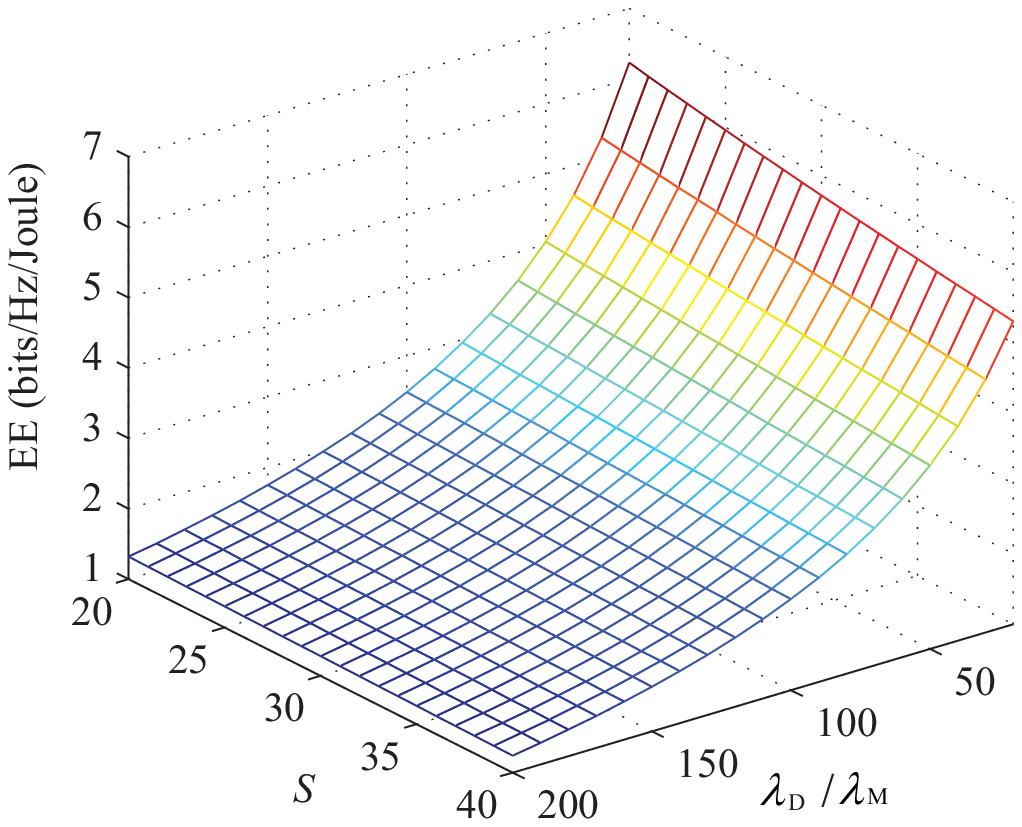}
}
\subfigure[EE of a D2D transmitter.]{
\label{Fig_d2d_ee_3D}
  \includegraphics[width=3.2 in, height=2.6 in]{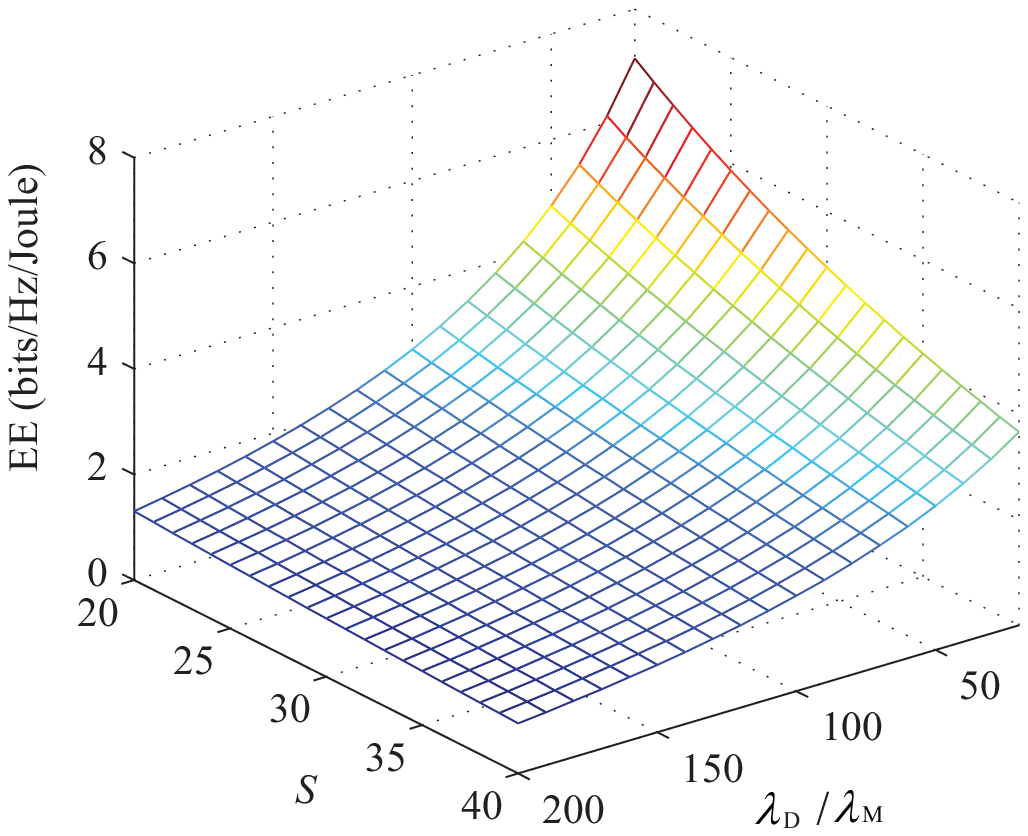}}
\caption{EE of uplink D2D underlaid massive MIMO cellular networks.}
\label{Fig6}
\end{figure}

In this subsection, we illustrate the interplay between massive MIMO and D2D. Specifically, massive MIMO allows MBS to accommodate more uplink information transmissions, and D2D links can be dense. Therefore, it is crucial to identify their combined effect. {In the simulations, we set  $d_o=50$ m, $\eta=0.9$, $N=400$, $P^\mathrm{D}_\mathrm{{max}}= 15$ dBm and ${I_\mathrm{th}} / {\sigma^2} = 0$ dB.}

Fig. \ref{Fig_mue_se_3D} shows the effects of different $S$ and D2D densities on area SE of the cellular. Serving more CUEs can improve the area SE of the cellular, due to the large multiplexing gains provided by massive MIMO. However, when D2D links grow large (e.g., $\lambda_\mathrm{D}=100 \times \lambda_\mathrm{M}$ in this figure.), increasing $S$ will not result in a big improvement of area SE. The reason is that D2D-to-cellular interference becomes severe in the dense D2D scenarios, which reduces the SE of a CUE.
Fig. \ref{Fig_d2d_se_3D} shows the effects of different $S$ and D2D densities on area SE of the D2D. We see that more cellular uplink transmissions will deteriorate the area SE of the D2D, due to the increase of cellular-to-D2D interference. More importantly, there exists the optimal D2D density value for maximizing the area SE of the D2D, beyond which, the area SE of the D2D decreases since a D2D user also suffers severe interference from other D2D transmissions.

Fig. \ref{Fig_mue_ee_3D} shows the effects of different $S$ and D2D densities on EE of a cellular user. We see that serving more CUEs will decrease the EE, which can be explained by two-fold: 1) The massive MIMO array gain allocated to each CUE decreases; and 2) the interference increases since there are more cellular transmissions. The D2D-to-cellular interference has a big adverse effect on the EE of a CUE. Fig. \ref{Fig_d2d_ee_3D} shows the effects of different $S$ and D2D densities on EE of a D2D transmitter. When more uplink transmissions are supported by massive MIMO aided cellular, the cellular-to-D2D interference increases, which has a detrimental effect on the EE of a D2D transmitter. The interference from other D2D transmissions also degrades the D2D performance. Moreover, it is indicated from Fig. \ref{Fig_mue_ee_3D} and Fig. \ref{Fig_d2d_ee_3D} that the EE of a CUE can be comparable to that of a D2D transmitter.

\section{Conclusion}
This paper took into account the uplink power control in the D2D underlaid massive MIMO cellular networks. The open-loop power control was  adopted to control the cellular user's transmit power, to mitigate the inter-cell and cellular-to-D2D interference. The D2D transmit power was controlled such that the average D2D signal power received by base stations is not larger than a certain value, to mitigate the D2D-to-cellular interference. We also considered the maximum transmit power constraints at the cellular users and D2D transmitters. We developed a tractable approach to provide the exact expressions for the area SE of the cellular and D2D tier. Two important properties were obtained, which can be viewed as sufficient conditions for satisfying the targeted SE. The average power consumption for a cellular user or D2D transmitter was derived to evaluate EE. Numerical results were presented to confirm the effectiveness of the proposed power control design.

\section*{Appendix A: A proof of Lemma 1}\label{App:A}

Based on \eqref{Power_control_D2D}, the cumulative distribution function (CDF) of ${P_{{\mathrm{D}}}}$  is written as
\begin{align}\label{pf_eq1}
& F_{P_{{\mathrm{D}}}}\left(x\right)=\Pr\left( {P_{{\mathrm{D}}}} \leq x \right) \nonumber\\
&=\Pr\left( \min\left\{\frac{I_\mathrm{th}}{L\left(\left|X_{\mathrm{D}, \mathrm{M}}\right|\right)},P_\mathrm{max}^\mathrm{D}\right\} \leq x \right) \nonumber\\
&=\left\{ \begin{array}{l}
1,\quad x \geq P_\mathrm{max}^\mathrm{D} \quad \\
\Delta\left(x\right),x < P_\mathrm{max}^\mathrm{D}
\end{array} \right. \nonumber\\
&=\mathrm{U}\left(x-P_\mathrm{max}^\mathrm{D}\right)\left(1-\Delta\left(x\right)\right)+\Delta\left(x\right),
\end{align}
where $\mathrm{U}\left(\cdot\right)$ is the unit step function denoted as $\mathrm{U}\left(x\right)=\left\{ \begin{array}{l}
1,\;\,x \ge 0\\
0,\;\,x < 0\quad
\end{array} \right.,$ and $\Delta\left(x\right)$ is calculated as
\begin{align}\label{pf_eq2}
 \Delta\left(x\right)&=\Pr \left( \frac{I_\mathrm{th}}{L{{\left( {\left| {X_{{\rm{D}},{\rm{M}}}} \right|} \right)}} } \leq x\right)\nonumber\\
 &=\Pr \left( {\left| {X_{{\rm{D}},{\rm{M}}}} \right| \leq {{\left( {\frac{{\beta x}}{I_\mathrm{th}}} \right)}^{1/{\alpha _{\rm{M}}}}}} \right).
\end{align}
Since the PDF of the distance $\left| {X_{{\rm{D}},{\rm{M}}}} \right|$ between a D2D transmitter and its nearest MBS is given by \cite{Han-Shin2012}
\begin{align}\label{pf_eq3}
f_{\left| {X_{{\rm{D}},{\rm{M}}}} \right|}\left(r\right)= 2 \pi \lambda_\mathrm{M} r \exp\left(-\pi \lambda_\mathrm{M}  r^2\right).
\end{align}
By using \eqref{pf_eq3}, \eqref{pf_eq2} is further derived as
\begin{align}\label{pf_eq4}
\Delta\left(x\right)&= \int_0^{{{\left( {\frac{{\beta x}}{I_\mathrm{th}}} \right)}^{1/{\alpha _{\rm{M}}}}}}  {{f_{\left| {X_{{\rm{D}},{\rm{M}}}} \right|}}\left( r \right)} dr \nonumber\\
&=1-\exp\left(-\pi \lambda_\mathrm{M} {\left( {\frac{\beta x}{I_\mathrm{th}}} \right)}^{2/{\alpha _{\rm{M}}}}\right).
\end{align}
Substituting \eqref{pf_eq4} into \eqref{pf_eq1}, we have
\begin{align}\label{delta_11}
\hspace{-0.5 cm} F_{P_{{\mathrm{D}}}}\left(x\right)=1-\mathrm{U}\left(P_\mathrm{max}^\mathrm{D}-x\right)\exp\Big(-\pi \lambda_\mathrm{M} {\left( {\frac{\beta x}{I_\mathrm{th}}} \right)}^{2/{\alpha _{\rm{M}}}}\Big).
\end{align}
Taking the derivative of $F_{P_{{\mathrm{D}}}}\left(x\right)$ in \eqref{delta_11}, we obtain the PDF of ${P_{{\mathrm{D}}}}$ in \eqref{PDF_D2D} and complete the proof.

\section*{Appendix B: A proof of Theorem 1}\label{App:B}

Based on \eqref{SINR_Macro}, the SE for a typical CUE is written as
\begin{align}\label{theo_1}
\overline{R}_\mathrm{C}&=\mathbb{E}\left\{\log_2\left(1+\mathrm{SINR}_\mathrm{M}\right)\right\} \nonumber\\
&=\mathbb{E}\left\{\log_2\left(1+\frac{Z_1}{I_\mathrm{M}+I_\mathrm{D}+ {\sigma^2}}\right)\right\},
\end{align}
where $Z_1={{P_{o,\mathrm{C}}}}{h_{o,{\mathrm{M}}}} L\left( {\left|{X_{o,{\mathrm{M}}}}\right|} \right)$.
Using \cite[Lemma 1]{Hamdi_2008}, \eqref{theo_1} can be equivalently transformed as
\begin{align}\label{theo_2}
&\overline{R}_\mathrm{C}=\frac{1}{\ln\left(2\right)} \mathbb{E}\left\{ \int_{\rm{0}}^\infty  {\frac{1}{t}} \left( {1 - {e^{ - t{Z_1}}}} \right){e^{ - t\left(I_\mathrm{M}+I_\mathrm{D}+ {\sigma^2}\right)}}dt \right\} \nonumber\\
&=\frac{1}{\ln\left(2\right)} \int_{\rm{0}}^\infty  {\frac{1}{t}} \underbrace{\mathbb{E}\left\{ \left( {1- e^{ - t{Z_1}}} \right)e^{ - t I_\mathrm{M}}  \right\}}_{\Xi_1\left(t\right)} \underbrace{\mathbb{E}\left\{e^{ - tI_\mathrm{D}}\right\} }_{\Xi_2\left(t\right)} e^{ - t\sigma^2} dt.
\end{align}
We first calculate $\Xi_1\left(t\right)$ as
\begin{align}\label{E_Z_1}
\Xi_1\left(t\right) &=\int_{\rm{0}}^\infty \mathbb{E}_{\left|{X_{o,{\mathrm{M}}}}\right|=x}\left\{ \left( {1- e^{ - t{Z_1}}} \right)e^{ - t I_\mathrm{M}}  \right\} f_{\left|{X_{o,{\mathrm{M}}}}\right|} \left(x\right) dx \nonumber\\
&\mathop \approx \limits^{(a)} \int_{\rm{0}}^\infty \left( {1- e^{ - t{{{P_{o,\mathrm{C}}}}(N-S+1) \beta x^{-\alpha_\mathrm{M}}} }} \right) \mathbb{E}_{\left|{X_{o,{\mathrm{M}}}}\right|=x}\left\{ e^{ - t I_\mathrm{M}}  \right\}  \nonumber\\
&\quad\quad\quad\times f_{\left|{X_{o,{\mathrm{M}}}}\right|} \left(x\right) dx,
\end{align}
where step (a) is obtained due to the fact that $h_{o,{\mathrm{M}}}\approx N-S+1$ for large $N$, $f_{\left|{X_{o,{\mathrm{M}}}}\right|} \left(x\right)$ is the PDF of the nearest distance between the typical CUE and its serving MBS, as seen in \eqref{pf_eq3}, and $\mathbb{E}_{\left|{X_{o,{\mathrm{M}}}}\right|=x}\left\{ e^{ - t I_\mathrm{M}}  \right\}$ in \eqref{E_Z_1} can be derived as
\begin{align}\label{E_I_M_further}
&\mathbb{E}_{\left|{X_{o,{\mathrm{M}}}}\right|=x}\left\{ e^{ - t I_\mathrm{M}}  \right\}\nonumber\\
&\mathop  = \limits^{(b)} \exp\left\{-S \lambda_\mathrm{M} \int_{{\mathcal{R}^{\rm{2}}}\backslash \mathrm{B}\left(o\right)} \left(1-\mathbb{E}\left\{e^{ - t  {P_{i,\mathrm{C}}}
{h_{i,\mathrm{M}}}\beta r^{-\alpha_\mathrm{M}} }\right\} \right) rdr  \right\} \nonumber\\
&=\exp\left\{-2\pi S \lambda_\mathrm{M} \int_x^\infty \left(1-\mathbb{E}\left\{e^{ - t  {P_{i,\mathrm{C}}}
{h_{i,\mathrm{M}}}\beta r^{-\alpha_\mathrm{M}} }\right\} \right) rdr  \right\} \nonumber\\
&\hspace{-0.4 cm}\mathop  = \limits^{(c)} \exp\bigg\{-2\pi S \lambda_\mathrm{M} \int_x^\infty \Big(1-\underbrace{\mathbb{E}_{{P_{i,\mathrm{C}}}}\left\{\frac{1}{1+t  {P_{i,\mathrm{C}}}
\beta r^{-\alpha_\mathrm{M}}} \right\}}_{\Upsilon_1} \Big)  rdr  \bigg\},
\end{align}
where step (b) is the generating functional of the PPP, and step (c) is given by considering ${h_{i,\mathrm{M}}}\sim \rm{exp}(1)$. Based on the power control given in \eqref{power_control_MBS}, $\Upsilon_1$ is given by
\begin{align}\label{Upsilon_1}
\Upsilon_1&=\left(1+t  {P_\mathrm{max}^{\mathrm{C}}}
\beta r^{-\alpha_\mathrm{M}}\right)^{-1}  \int_{r_o}^\infty  f_{\left|{X_{i,{\mathrm{M}}}}\right|} \left(\nu\right) d\nu \nonumber\\
&+\int_0^{r_o} \left(1+t P_o \beta^{1-\eta} \nu^{\eta \alpha_\mathrm{M}}
 r^{-\alpha_\mathrm{M}}\right)^{-1} f_{\left|{X_{i,{\mathrm{M}}}}\right|} \left(\nu\right) d\nu,
\end{align}
where $r_o=\left( {\frac{{{P_{{\rm{max}}}^\mathrm{C}}}}{{{P_o}}}} \right)^{1/({\alpha_\mathrm{M}}\eta) }{\beta^{1/{\alpha_\mathrm{M}}}}$ represents the distance such that the path loss compensation reaches the maximum value under power constraint, and $f_{\left|{X_{i,{\mathrm{M}}}}\right|} \left(\nu\right)$ is the PDF of the nearest distance between the interfering CUE $i$ and its serving MBS. Substituting \eqref{Upsilon_1} and \eqref{E_I_M_further} into \eqref{E_Z_1}, after some manipulations, we obtain \eqref{Xi_1}.
Then, we have
\begin{align}\label{E_I_M}
&\Xi_2\left(t\right)=\mathbb{E}\left\{e^{ - t \sum\nolimits_{j \in {\Phi_\mathrm{D}}} {{P_{j,{\mathrm{D}}}}
{h_{j,\mathrm{M}}}L\left( {\left|X_{j,\mathrm{M}}\right|}\right)}}\right\}\nonumber\\
&=\exp\Big\{-2 \pi \lambda_\mathrm{D} \int_{0}^\infty \Big(1-\mathbb{E}
\Big\{e^{-t {{P_{j,{\mathrm{D}}}}
{h_{j,\mathrm{M}}}L\left( {\left|X_{j,\mathrm{M}}\right|}\right)}}  \Big\}\Big)r dr
\Big\}.
\end{align}
After some manipulations, the above can be derived as~\cite{Haenggi2009}
\begin{align}\label{E_I_M_1}
\Xi_2\left(t\right)=&\exp\Big(-\pi \lambda_\mathrm{D} \beta^{\frac{2}{\alpha_\mathrm{M}}} \mathbb{E}\left\{({P_{j,{\mathrm{D}}}})^{\frac{2}{\alpha_\mathrm{M}}}\right\} \nonumber\\
& \times \Gamma(1+\frac{2}{\alpha_\mathrm{M}}) \Gamma(1-\frac{2}{\alpha_\mathrm{M}}) t^\frac{2}{\alpha_\mathrm{M}}  \Big),
\end{align}
where $\Gamma(\cdot)$ is the Gamma function~\cite{gradshteyn}. By using \textbf{Lemma 1},  $\mathbb{E}\left\{({P_{j,{\mathrm{D}}}})^{\frac{2}{\alpha_\mathrm{M}}}\right\}$ is given by
\begin{align}\label{power_D2D_exp}
&\mathbb{E}\left\{({P_{j,{\mathrm{D}}}})^{\frac{2}{\alpha_\mathrm{M}}}\right\}=\int_0^\infty x^{\frac{2}{\alpha_\mathrm{M}}} f_{P_{{\mathrm{D}}}} \left(x\right)  dx \nonumber\\
&=\frac{{(P_\mathrm{max}^\mathrm{D})^{\frac{2}{{\alpha _{\rm{M}}}}}}}{\varpi_0}(1-e^{-\varpi_0}-\varpi_0 e^{-\varpi_0})+ \left(P_\mathrm{max}^\mathrm{D}\right)^{\frac{2}{\alpha_\mathrm{M}}} {\bar{\Delta}\left(P_\mathrm{max}^\mathrm{D}\right)},
\end{align}
where $f_{P_{{\mathrm{D}}}} \left(x\right)$ is given by \eqref{PDF_D2D}, $\varpi_0=\pi \lambda_\mathrm{M}{\left( \frac{\beta P_\mathrm{max}^\mathrm{D}}{I_{{\rm{th}}}} \right)^{2/{\alpha _{\rm{M}}}}}$. Substituting \eqref{power_D2D_exp} into \eqref{E_I_M_1}, we obtain \eqref{Xi_2}.

\section*{Appendix C: A Proof of Scaling Property 1}\label{App:C}
 Based on \eqref{theo_1}, the SE for a CUE can be tightly lower bounded as~\cite{Lifeng_massiveMIMO}
\begin{align}\label{cue_capacity_LB_p}
\overline{R}^\mathrm{L}_\mathrm{C}= {\log _2}\left( {1 + {X_1 e^{ {Y_2}}}} \right),
\end{align}
where $X_1=e^{Y_1}$, and
\begin{equation}\label{X_1_X_2}
\left\{ \begin{array}{l}
{Y_1} = {\rm{E}}\left\{ {\ln \left( {{P_{o,\mathrm{C}}}{h_{o,\mathrm{M}}}\beta {{\left| {{X_{o,\mathrm{M}}}} \right|}^{ - {\alpha_\mathrm{M}}}}} \right)} \right\}\\
{Y_2} = {\rm E}\left\{ {\ln \left( {\frac{1}{{{I_\mathrm{M}} + {I_\mathrm{D}} + {\sigma ^2}}}} \right)} \right\}
\end{array} \right.
\end{equation}

We first calculate $Y_1$ as
\begin{align}
{Y_1} =&{\rm{E}}\left\{ {\ln \left( {{P_{o,\mathrm{C}}}} \right)} \right\} - {\alpha_\mathrm{M}}{\rm E}\left\{ {\ln \left( {\left| {{X_{o,\mathrm{M}}}} \right|} \right)} \right\} \nonumber\\
&\quad+ {\rm{E}}\left\{ {\ln \left( {{h_{o,\mathrm{M}}}} \right)} \right\} + \ln \left( \beta  \right).
\end{align}

Based on the uplink power control given in \eqref{power_control_MBS}, we obtain ${\rm{E}}\left\{ {\ln \left( {{P_{o,\mathrm{C}}}} \right)} \right\}$ as
\begin{align}
&  {\rm{E}}\left\{ {\ln \left( {{P_{o,\mathrm{C}}}} \right)} \right\} = \int_0^\infty  {{{\rm E}_{\left| {{X_{o,\mathrm{M}}}} \right| = x}}\left\{ {\ln \left( {{P_{o,\mathrm{C}}}} \right)} \right\}} {f_{\left| {{X_{o,\mathrm{M}}}} \right|}}\left( x \right)dx  \nonumber \\
& = \int_0^{{r_o}} {\left( {\ln \left( {{P_o}{\beta ^{ - \eta }}} \right) + \eta {\alpha_\mathrm{M}}\ln \left( x \right)} \right){f_{\left| {{X_{o,\mathrm{M}}}} \right|}}\left( x \right)dx}  \nonumber \\
&\qquad\qquad+ \int_{{r_o}}^\infty  {\ln \left( {{P^\mathrm{C}_{\max }}} \right){f_{\left| {{X_{o,\mathrm{M}}}} \right|}}\left( x \right)dx} \nonumber \nonumber\\
&  = \ln \left( {{P_o}{\beta ^{ - \eta }}} \right)\left( {1 - \exp \left( { - \pi {\lambda_\mathrm{M}} r_o^2} \right) } \right) + \ln \left( {{P^\mathrm{C}_{\max }}} \right)\exp \left( { - \pi {\lambda_\mathrm{M}} r_o^2} \right) \nonumber \\
& + \frac{\eta {\alpha_\mathrm{M}} }{2}\left(\ell +\Gamma\left(0,r_o^2 \pi {\lambda_\mathrm{M}}\right)+2e^{-r_o^2 \pi {\lambda_\mathrm{M}}}\ln(r_o)+\ln(\pi {\lambda_\mathrm{M}})\right),
\end{align}
where $\ell \approx 0.5772$.
%

We then derive ${\rm E}\left\{ {\ln \left( {\left| {{X_{o,\mathrm{M}}}} \right|} \right)} \right\}$ as
\begin{align}
 {\rm E}\left\{ {\ln \left( {\left| {{X_{o,\mathrm{M}}}} \right|} \right)} \right\} & = \int_0^\infty  {\ln \left( x \right)} {f_{\left| {{X_{o,\mathrm{M}}}} \right|}}\left( x \right)dx \nonumber \\
&=\int_0^\infty  {\ln \left( x \right)}  2 \pi \lambda_\mathrm{M} x \exp\left(-\pi \lambda_\mathrm{M} x^2\right)    dx \nonumber \\
&  = \frac{1}{2}\psi \left( 1 \right) - \frac{1}{2}\ln \left( {\pi {\lambda _M}} \right).
\end{align}

Considering that $h_{o,{\mathrm{M}}}\sim \Gamma\left(N-S+1,1\right)$, we have ${\rm{E}}\left\{ {\ln \left( {{h_{o,\mathrm{M}}}} \right)} \right\}=\psi \left( {N - S + 1} \right)$. Thus, we can obtain $X_1=e^{Y_1}$  given in \eqref{X_1_theo}.

By using Jensen's inequality, we can derive the lower bound on the $Y_2$ as
\begin{align}\label{Jensen_inequa}
{Y_2} > {{\overline Y}_2} = \ln \left( {\frac{1}{{{\rm E}\left\{ {{I_\mathrm{M}}} \right\} + {\rm E}\left\{ {{I_\mathrm{D}}} \right\} + {\sigma ^2}}}} \right).
\end{align}
We first have
\begin{align}\label{Exp_I_M}
{\rm E}\left\{ {{I_\mathrm{M}}} \right\} =\int_0^\infty  {{{\rm{E}}_{\left| {{X_{o,\mathrm{M}}}} \right| = x}}\left\{ {{I_\mathrm{M}}} \right\}} {f_{\left| {{X_{o,\mathrm{M}}}} \right|}}\left( x \right)dx,
\end{align}
where ${{\rm{E}}_{\left| {{X_{o,\mathrm{M}}}} \right| = x}}\left\{ {{I_\mathrm{M}}} \right\}$ is given by
\begin{align}
& {{\rm{E}}_{\left| {{X_{o,\mathrm{M}}}} \right| = x}}\left\{ {{I_\mathrm{M}}} \right\} = {\rm E}\left\{ {\sum\nolimits_{i \in {\Phi _{u,\mathrm{M}\backslash B\left( o \right)}}} {{P_{i,\mathrm{C}}}{h_{i,\mathrm{M}}}\beta {r^{ - {\alpha_\mathrm{M}}}}} } \right\} \nonumber \\
& \mathop  = \limits^{\left(a\right)}\beta 2\pi S{\lambda_\mathrm{M}} \mathbb{ E}\left\{{P_{i,\mathrm{C}}}\right\}  \frac{x^{2- {\alpha_\mathrm{M}}}}{\alpha_\mathrm{M}-2},
\end{align}
in which $\left(a\right)$ results from Campbell's theorem, and the average transmit power of the CUE is calculated as
\begin{align}\label{average_transmit_power_CUE_112}
&\mathbb{ E}\left\{{P_{i,\mathrm{C}}}\right\}= \int_0^\infty  {{{\rm E}_{\left| {{X_{i,\mathrm{M}}}} \right| = x}}\left\{ {{{P_{i,C}}}} \right\}} {f_{\left| {{X_{i,\mathrm{M}}}} \right|}}\left( x \right)dx \nonumber \\
& = \int_0^{{r_o}} {\left(P_o \beta^{-\eta}{ x^{\eta \alpha_\mathrm{M}}} \right){f_{\left| {{X_{i,\mathrm{M}}}} \right|}}\left( x \right)dx}  + \int_{{r_o}}^\infty  {{{P^C_{\max }}}{f_{\left| {{X_{i,\mathrm{M}}}} \right|}}\left( x \right)dx} \nonumber \\
& = P_o \beta^{-\eta} {\left( {\pi {\lambda_\mathrm{M}}} \right)^{ - \frac{{\eta {\alpha_\mathrm{M}}}}{2}}}\left( {\Gamma \left( {1 + \frac{{\eta {\alpha_\mathrm{M}}}}{2}} \right) - \Gamma \left( {1 + \frac{{\eta {\alpha _\mathrm{M}}}}{2},\pi {\lambda_\mathrm{M}}\sqrt {{r_o}} } \right)} \right) \nonumber\\
&\qquad\qquad+ {{P^C_{\max }}} \exp \left( { - \pi {\lambda_\mathrm{M}} r_o^2} \right).
\end{align}
Likewise, ${ \mathbb{E}}\left\{ {{I_\mathrm{D}}} \right\}$ is derived as
\begin{align}\label{Exp_I_D}
{ \mathbb{E}}\left\{ {{I_\mathrm{D}}} \right\} &=2\pi {\lambda_\mathrm{D}}\beta {\mathbb{ E}}\left\{ {{P_{j,\mathrm{D}}}} \right\} \int_{0}^\infty \left(\max\left(D_o,r\right)\right)^{ - {\alpha_\mathrm{M}}} r dr \nonumber\\
&=2\pi {\lambda_\mathrm{D}}\beta {\mathbb{ E}}\left\{ {{P_{j,\mathrm{D}}}} \right\} \left(\frac{D_o^{2 - {\alpha_\mathrm{M}}}}{2} +\frac{D_o^{2 - {\alpha_\mathrm{M}}}}{{\alpha_\mathrm{M}}-2}\right)
\end{align}
where $D_o$ is the minimum distance between a D2D transmitter and the typical serving MBS in practice, and ${\rm E}\left\{ {{P_{j,\mathrm{D}}}} \right\}$ is given by
\begin{align}
\label{meanD2DTP}
{\rm E}\left\{ {{P_{j,\mathrm{D}}}} \right\} & = \int_0^\infty \bar{F}_{P_{j,\mathrm{D}}}\left(x\right) dx,
\end{align}
where $\bar{F}_{P_{j,\mathrm{D}}}\left(x\right)$ is the complementary cumulative distribution function. Based on \eqref{PDF_D2D}, we have
\begin{align}\label{CCDF_F_j_D}
&\bar{F}_{P_{j,\mathrm{D}}}\left(x\right)=\Pr\left(P_{j,\mathrm{D}}>x\right) \nonumber\\
&=\Pr\left( \min\left\{\frac{I_\mathrm{th}}{L\left(\left|X_{\mathrm{D}, \mathrm{M}}\right|\right)},P_\mathrm{max}^\mathrm{D}\right\} > x \right) \nonumber\\
&=\left\{ \begin{array}{l}
0,\quad x \geq P_\mathrm{max}^\mathrm{D} \quad \\
\widetilde{\Delta}\left(x\right),x < P_\mathrm{max}^\mathrm{D}
\end{array} \right.,
\end{align}
where $\widetilde{\Delta}\left(x\right)$ is
\begin{align}\label{Delta_line}
&\widetilde{\Delta}\left(x\right)=\Pr\left(\frac{I_\mathrm{th}}{L\left(\left|X_{\mathrm{D}, \mathrm{M}}\right|\right)}> x \right) \nonumber\\
&={\rm{\mathbf{1}}}\left( x < \frac{D_o^{\alpha _{\rm{M}}}I_\mathrm{th}}{\beta}  \right)  \int_0^{D_o} f_{\left| {X_{{\rm{D}},{\rm{M}}}} \right|}\left(r\right) dr+ \int_{\varpi_1}^\infty f_{\left| {X_{{\rm{D}},{\rm{M}}}} \right|}\left(r\right) dr \nonumber\\
&={\rm{\mathbf{1}}}\left( x < \frac{D_o^{\alpha _{\rm{M}}}I_\mathrm{th}}{\beta}  \right) \left(1-\exp\left(-\pi \lambda_\mathrm{M} D_o^2\right) \right)\nonumber\\
&~~+\exp\left(-\pi \lambda_\mathrm{M} (\varpi_1(x))^2\right)
\end{align}
where $\varpi_1(x)=\max\left\{D_o,{{\left( {\frac{{\beta x}}{I_\mathrm{th}}} \right)}^{1/{\alpha _{\rm{M}}}}}\right\}$. By substituting \eqref{CCDF_F_j_D} into \eqref{meanD2DTP}, ${\rm E}\left\{ {{P_{j,\mathrm{D}}}} \right\}$ is derived as
\begin{align}\label{Exp_P_j_D}
{\rm E}\left\{ {{P_{j,\mathrm{D}}}} \right\}=\int_0^{P_\mathrm{max}^\mathrm{D}} \widetilde{\Delta}\left(x\right)  dx.
\end{align}
Substituting \eqref{Exp_I_M} and \eqref{Exp_I_D} into the right-hand-side of \eqref{Jensen_inequa}, we obtain ${{\overline Y}_2}$. According to \eqref{cue_capacity_LB_p}, the expected $\overline{R}_\mathrm{C}^\mathrm{th}$ can be satisfied  when $\overline{R}_\mathrm{C}^\mathrm{th}\leq {\log _2}\left( {1 + {X_1 e^{ {{\overline Y}_2}}}} \right)$. Therefore, we have ${{\overline  Y}_2} \geq \ln \left( {\frac{{{2^{\bar R_{\rm{C}}^{{\rm{th}}}}} - 1}}{{{X_1}}}} \right)$, after some manipulations, we obtain the desired result given in \eqref{corollary_1}.

\section*{Appendix D: A proof of Theorem 2}\label{App:D}

Based on \eqref{SINR_Small}, $\overline{R}_\mathrm{D}$ is given by
\begin{align}\label{theo_2_1}
&\overline{R}_\mathrm{D}=\mathbb{E}\left\{\log_2\left(1+\mathrm{SINR}_\mathrm{D}\right)\right\} \nonumber\\
&\mathop  = \limits^{(a)}1/\ln(2) \int_0^\infty  \frac{1}{t}\Big(1 -\underbrace{\mathbb{E}\left\{ {{e^{ - t{Z_2}}}} \right\}}_{{\Xi_3}\left(t\right)}\Big) \underbrace{\mathbb{E}\left\{ {{e^{ - t{\left(J_\mathrm{M}+J_\mathrm{D}\right)}}}} \right\}}_{\Xi_4\left(t\right)}{e^{ - t{\sigma ^2}}}dt,
\end{align}
where $Z_2={{P_{o,\mathrm{D}}}{g_{o,\mathrm{D}}}L\left( d_o \right)}$, and step (a) is obtained by following the similar approach in \eqref{theo_2}. Using \textbf{Lemma } 1, we first derive ${\Xi_3}\left(t\right)$ as
\begin{align}\label{Xi_2_theo}
&{\Xi_3}\left(t\right)=\mathbb{E}_{P_{o,\mathrm{D}}}\left\{\mathbb{E}_{g_{o,\mathrm{D}}}\left\{ {{e^{ - t{Z_2}}}} \right\}\right\}=\mathbb{E}_{P_{o,\mathrm{D}}}\left\{\frac{1}{{1 + t{P_{o,\mathrm{D}}} \beta d_o^{-\alpha_\mathrm{D}}}}\right\} \nonumber\\
&=\int_0^\infty \left({{1 + t x \beta d_o^{-\alpha_\mathrm{D}}}}\right)^{-1} f_{{P_{{\mathrm{D}}}}} \left(x\right) dx.
\end{align}
Substituting \eqref{PDF_D2D} into \eqref{Xi_2_theo}, we obtain \eqref{Xi_3}.

Considering that the cellular interference $J_\mathrm{M}$ and D2D interference $J_\mathrm{D}$ are independent, $\Xi_4\left(t\right)$ is calculated as
\begin{align}\label{Xi_4_step1}
\Xi_4\left(t\right)=\mathbb{E}\left\{ e^{ - t J_\mathrm{M}} \right\}\mathbb{E}\left\{ e^{ - t J_\mathrm{D}} \right\}.
\end{align}
Similar to \eqref{E_I_M_1}, $\mathbb{E}\left\{ e^{ - t J_\mathrm{M}} \right\}$ is derived as
\begin{align}\label{J_M_1}
\hspace{-0.4 cm}\mathbb{E}\left\{ e^{ - t J_\mathrm{M}} \right\}&=\exp\Big(-\pi S \lambda_\mathrm{M} \beta^{\frac{2}{\alpha_\mathrm{D}}} \Omega_1 \Gamma(1+\frac{2}{\alpha_\mathrm{D}}) \Gamma(1-\frac{2}{\alpha_\mathrm{D}}) t^\frac{2}{\alpha_\mathrm{D}}  \Big),
\end{align}
where $\Omega_1=\mathbb{E}\left\{({P_{i,{\mathrm{C}}}})^{\frac{2}{\alpha_\mathrm{D}}}\right\}$ is given by
\begin{align}\label{E_P_i_C}
\Omega_1&=\int_0^{r_o} \big(P_o \beta^{-\eta} \nu^{\eta \alpha_\mathrm{M} }  \big)^{\frac{2}{\alpha_\mathrm{D}}} f_{\left|{X_{i,{\mathrm{M}}}}\right|} \left(\nu\right) d\nu \nonumber\\
&\qquad+\big(P_\mathrm{max}^{\mathrm{C}}\big)^{\frac{2}{\alpha_\mathrm{D}}}\int_{r_o}^\infty
f_{\left|{X_{i,{\mathrm{M}}}}\right|} \left(\nu\right) d\nu.
\end{align}
After some manipulations, we derive $\Omega_1$ as \eqref{Omega_123}. Likewise, $\mathbb{E}\left\{ e^{ - t J_\mathrm{D}} \right\}$ is derived as
\begin{align}\label{E_J_gamma}
\mathbb{E}\left\{ e^{ - t J_\mathrm{D}} \right\}= &\exp\Big(-\pi \lambda_\mathrm{D} \beta^{\frac{2}{\alpha_\mathrm{D}}} \Omega_2
 \Gamma(1+\frac{2}{\alpha_\mathrm{D}}) \Gamma(1-\frac{2}{\alpha_\mathrm{D}}) t^\frac{2}{\alpha_\mathrm{D}}  \Big),
\end{align}
where  $\Omega_2=\mathbb{E}\left\{({P_{j,{\mathrm{D}}}})^{\frac{2}{\alpha_\mathrm{D}}}\right\}$, using \textbf{Lemma} 1, we have
\begin{align}\label{E_P_j_D}
&\Omega_2=\int_0^{P_\mathrm{max}^\mathrm{D}}  x^{\frac{2}{\alpha_\mathrm{D}}} f_{P_{{\mathrm{D}}}} \left(x\right)  dx +
\left(P_\mathrm{max}^\mathrm{D}\right)^{\frac{2}{\alpha_\mathrm{D}}} {\bar{\Delta}\left(P_\mathrm{max}^\mathrm{D}\right)}.
\end{align}
After some manipulations, we derive $\Omega_2$ as \eqref{Omega_223} at the following page. Then,
we attain \eqref{Xi_4} by substituting \eqref{J_M_1} and \eqref{E_J_gamma} into \eqref{Xi_4_step1}.

\section*{Appendix E:A proof of Scaling Property 2}\label{App:E}

Similar to \eqref{theo_2_1}, the SE for a D2D link can be tightly lower bounded as
\begin{equation}
\label{due_capacity_LB_p}
{\rm{\bar R}}_{^{\rm{D}}}^{\rm{L}} = {\log _2}\left( {1 + {X_4}{e^{{Y_4}}}} \right)
\end{equation}
where $ {X_4} = {e^{{Y_3}}}$, and
\[\left\{ \begin{array}{l}
{Y_3} = {\rm E}\left\{ {\ln \left( {{P_{o,{\rm{D}}}}{g_{o,{\rm{D}}}}\beta {d_o}^{ - {\alpha _{\rm{D}}}}} \right)} \right\}\\
{Y_4} = {\rm E}\left\{ {\ln \left( {\frac{1}{{{J_{\rm{M}}} + {J_{\rm{D}}} + {\sigma ^2}}}} \right)} \right\}
\end{array} \right.\]

We first calculate $ Y_3 $ as
\begin{align}
{Y_3} = {\rm E}\left\{ {\ln \left( {{P_{o,{\rm{D}}}}} \right)} \right\} + {\rm E}\left\{ {\ln \left( {{g_{o,{\rm{D}}}}} \right)} \right\} - {\alpha _{\rm{D}}}\ln \left( {{d_o}} \right) + \ln \left( \beta  \right)
\end{align}

Based on the D2D power control in \eqref{Power_control_D2D} and \textbf{Lemma} \ref{Lemma1}, we obtain $  {\rm E}\left\{ {\ln \left( {{P_{o,{\rm{D}}}}} \right)} \right\}  $ as
\begin{align}
&{\rm E}\left\{ {\ln \left( {{P_{o,{\rm{D}}}}} \right)} \right\} = \int_0^\infty  {\ln \left( x \right){f_{{P_{\rm{D}}}}}\left( x \right)dx} \nonumber \\
& = \int_0^{P_{\max }^{\rm{D}}} {\ln \left( x \right){f_{{P_{\rm{D}}}}}\left( x \right)dx + \ln \left( {P_{\max }^{\rm{D}}} \right)} \bar \Delta \left( {P_{\max }^{\rm{D}}} \right)
\end{align}
Considering that ${{g_{o,{\rm{D}}}}} \backsim \exp (1)$, we have ${\rm E}\left\{ {\ln \left( {{g_{o,{\rm{D}}}}} \right)} \right\} =\int_0^\infty \ln(x)e^{-x}dx=\ell \approx 0.5772$. Thus, we can obtain ${X_4} = {e^{{Y_3}}} $ given in \eqref{X_4_theo}.

By using Jensen's inequality, we can derive the lower bound on the $ Y_4 $ as
\begin{align}\label{bar_Y_4}
{Y_4} > {{\overline  Y}_4} = \ln \left( {\frac{1}{{{\rm E}\left\{ {{J_{\rm{M}}}} \right\} + {\rm E}\left\{ {{J_{\rm{D}}}} \right\} + {\sigma ^2}}}} \right).
\end{align}
We first derive ${\rm E}\left\{ {{J_{\rm{M}}}} \right\} $ as
\begin{align}\label{APP_D_E_j_M}
& {\rm E}\left\{ {{J_{\rm{M}}}} \right\} = {\rm E}\left\{ {\sum\nolimits_{i \in {\Phi _{u,{\rm{M}}}}} {{P_{i,C}}{g_{i,{\rm{D}}}}L\left( {\left| {{X_{i,{\rm{D}}}}} \right|} \right)} } \right\} \nonumber \\
& = 2\pi S{\lambda _M}{\rm E}\left\{ {{P_{i,{\rm{C}}}}} \right\}\beta \int_0^\infty  {{\left(\max\left(D_1,r\right)\right)^{ - {\alpha _{\rm{D}}}}}r dr}\nonumber\\
&=2\pi S{\lambda _M}{\rm E}\left\{ {{P_{i,{\rm{C}}}}} \right\}\beta \left(\frac{D_1^{2 - {\alpha_\mathrm{D}}}}{2} +\frac{D_1^{2 - {\alpha_\mathrm{M}}}}{{\alpha_\mathrm{D}}-2}\right),
\end{align}
where $ {\rm E}\left\{ {{P_{i,{\rm{C}}}}} \right\} $ is given by \eqref{average_transmit_power_CUE_112}, and $D_1$ is the reference distance to avoid singularity.

Similar to \eqref{Exp_I_D}, $ {\rm E}\left\{ {{J_{\rm{D}}}} \right\} $  is calculated as
\begin{align}\label{D2D_intefere_APP_D}
{\rm E}\left\{ {{J_{\rm{D}}}} \right\}=2\pi {\lambda_\mathrm{D}}\beta {\mathbb{ E}}\left\{ {{P_{j,\mathrm{D}}}} \right\} \left(\frac{D_2^{2 - {\alpha_\mathrm{D}}}}{2} +\frac{D_2^{2 - {\alpha_\mathrm{D}}}}{{\alpha_\mathrm{D}}-2}\right),
\end{align}
where ${\mathbb{ E}}\left\{ {{P_{j,\mathrm{D}}}} \right\}$ is given by \eqref{Exp_P_j_D}, and $D_2$ is the reference distance.

Substituting \eqref{APP_D_E_j_M} and \eqref{D2D_intefere_APP_D} into the right-hand-side of \eqref{bar_Y_4},
we obtain ${{\overline Y}_4}$. Based on \eqref{due_capacity_LB_p} and \eqref{bar_Y_4}, $\overline{R}_\mathrm{D}^\mathrm{th}\leq {\log _2}\left( {1 + {X_4 e^{ {{\overline Y}_4}}}} \right) \Rightarrow {{\overline  Y}_4} \geq \ln \left( {\frac{{{2^{\bar
R_{\rm{D}}^{{\rm{th}}}}} - 1}}{{{X_4}}}} \right)$, after some manipulations, we obtain  \eqref{corollary_2} and complete the proof.

\bibliographystyle{IEEEtran}

\end{document}